\documentclass[11pt,a4paper]{article}
\usepackage{amssymb}
\usepackage{amsmath}
\usepackage{amsfonts}
\usepackage{dsfont}
\usepackage{amsthm}
\usepackage{mathrsfs}
\usepackage{hyperref}
\usepackage{color}
\usepackage[margin=2.41cm]{geometry}
\usepackage[all,cmtip]{xy}
\usepackage[utf8]{inputenc}
\usepackage{graphicx}
\usepackage{varwidth}
\usepackage{comment}

\usepackage{upgreek}
\usepackage{rotating}

\usepackage{tikz}
\usetikzlibrary{shapes.geometric}

\usepackage{tikz-cd}
\usetikzlibrary{matrix,arrows,calc,decorations.pathmorphing,fit,shapes.geometric,decorations.pathreplacing,positioning}

\usepackage[shortlabels]{enumitem}

\definecolor{darkred}{rgb}{0.8,0.1,0.1}
\hypersetup{
     colorlinks=false,         
     linkcolor=darkred,
     citecolor=blue,
}

\theoremstyle{plain}
\newtheorem{theo}{Theorem}[section]
\newtheorem{lem}[theo]{Lemma}
\newtheorem{propo}[theo]{Proposition}
\newtheorem{cor}[theo]{Corollary}

\theoremstyle{definition}
\newtheorem{defi}[theo]{Definition}

\newtheorem{exx}[theo]{Example}
\newenvironment{ex}
  {\pushQED{\qed}\exx}
  {\popQED\endexx}
  
\newtheorem{cexx}[theo]{Counterexample}
\newenvironment{cex}
  {\pushQED{\qed}\cexx}
  {\popQED\endcexx}

\newtheorem{remm}[theo]{Remark}
\newenvironment{rem}
  {\pushQED{\qed}\remm}
  {\popQED\endremm}

\numberwithin{equation}{section}

\def\nn{\nonumber}

\def\bbK{\mathbb{K}}
\def\bbR{\mathbb{R}}

\def\bbZ{\mathbb{Z}}

\def\bbL{\mathbb{L}}

\def\Hom{\mathrm{Hom}}

\def\id{\mathrm{id}}

\def\supp{\mathrm{supp}}

\def\dd{\mathrm{d}}

\def\1{\mathbf{1}}
\def\oone{\mathds{1}}
\def\op{\mathrm{op}}

\def\Loc{\mathbf{Loc}}
\def\Lan{\operatorname{Lan}}

\def\AQFT{\mathbf{AQFT}}
\def\2AQFT{\mathbf{2AQFT}}

\def\Fun{\mathbf{Fun}}

\def\Disk{\mathbf{Disk}}

\def\Alg{\mathbf{Alg}}

\def\Vec{\mathbf{Vec}}
\def\Ch{\mathbf{Ch}}

\def\BB{\mathbf{B}}
\def\CC{\mathbf{C}}
\def\DD{\mathbf{D}}

\def\EE{\mathbf{E}}

\def\TT{\mathbf{T}}
\def\Cat{\mathbf{Cat}}
\def\OCat{\mathbf{Cat}^{\perp}}

\def\sSet{\mathbf{sSet}}

\def\Op{\mathbf{Op}}

\def\Tot{\mathrm{Tot}}

\def\AAA{\mathfrak{A}}
\def\BBB{\mathfrak{B}}

\def\BBB{\mathfrak{B}}

\def\FFF{\mathfrak{F}}
\def\GGG{\mathfrak{G}}

\def\O{\mathcal{O}}
\def\P{\mathcal{P}}
\def\Q{\mathcal{Q}}

\def\t{\mathsf{t}}

\def\Mor{\mathrm{Mor}}

\newcommand\und[1]{\underline{#1}}
\newcommand\ovr[1]{\overline{#1}}

\DeclareMathOperator*{\Motimes}{\text{\raisebox{0.25ex}{\scalebox{0.8}{$\bigotimes$}}}}

\def\sk{\vspace{1mm}}

\makeatletter
\let\@fnsymbol\@alph
\makeatother

\title{%
Strictification theorems for the homotopy time-slice axiom
}

\author{%
Marco Benini$^{1,2,a}$, Victor Carmona$^{3,b}$\ and\ 
Alexander Schenkel$^{4,c}$\vspace{4mm}\\
{\small ${}^1$ Dipartimento di Matematica, Universit\`a di Genova,}\\
{\small Via Dodecaneso 35, 16146 Genova, Italy.}\vspace{2mm}\\
{\small ${}^2$ INFN, Sezione di Genova,}\\
{\small Via Dodecaneso 33, 16146 Genova, Italy.}\vspace{2mm}\\
{\small ${}^3$ Departamento de {\'A}lgebra, Universidad de Sevilla,}\\
{\small Avda Reina Mercedes s/n, 41012 Sevilla, Spain.}\vspace{2mm}\\
{\small ${}^4$ School of Mathematical Sciences, University of Nottingham,}\\
{\small University Park, Nottingham NG7 2RD, United Kingdom.}\vspace{4mm}\\
{\small \begin{tabular}{ll}
Email: & ${}^a$~\texttt{benini@dima.unige.it}\\
& ${}^b$~\texttt{vcarmona1@us.es}\\
& ${}^c$~\texttt{alexander.schenkel@nottingham.ac.uk}\vspace{2mm}
\end{tabular}
}
}

\date{February 2023}

\begin{document}

\maketitle

\begin{abstract}
\noindent It is proven that the homotopy time-slice axiom for many types of algebraic quantum field theories (AQFTs) taking values in chain complexes can be strictified. This includes the cases of Haag-Kastler-type AQFTs on a fixed globally hyperbolic Lorentzian manifold (with or without time-like boundary), locally covariant conformal AQFTs in two spacetime dimensions, locally covariant AQFTs in one spacetime dimension, and the relative Cauchy evolution. The strictification theorems established in this paper prove that, under suitable hypotheses that hold true for the examples listed above, there exists a Quillen equivalence between the model category of AQFTs satisfying the homotopy time-slice axiom and the model category of AQFTs satisfying the usual strict time-slice axiom.
\end{abstract}



\renewcommand{\baselinestretch}{0.8}\normalsize
\tableofcontents
\renewcommand{\baselinestretch}{1.0}\normalsize

\newpage


\section{\label{sec:intro}Introduction and summary}
The time-slice axiom is one of the central axioms of algebraic quantum field theory (AQFT).
It introduces a notion of time evolution on globally hyperbolic
Lorentzian manifolds, which is the key ingredient for analyzing the
physical behavior of an AQFT, for instance through the relative Cauchy evolution (RCE) 
and its associated stress-energy tensor \cite{BFV,FewsterVerch}, 
or through the local measurement schemes introduced in \cite{FewsterVerchmeasurement}.
\sk

At a more technical level, an AQFT is described by an algebra $\AAA : \O_{\ovr{\CC}}\to \TT$
over a suitable colored operad $\O_{\ovr{\CC}}$, called the AQFT operad \cite{BSWoperad}, that is
associated to a category of spacetimes $\ovr{\CC}$. (See Subsection 
\ref{subsec:AQFT} for the relevant mathematical background.) The symmetric monoidal target
category $\TT$ is arbitrary, but the traditional examples of AQFTs are based on 
the closed symmetric monoidal category $\TT=\Vec_\bbK$ of vector spaces over a field $\bbK$ of characteristic $0$.
In this context the time-slice axiom is implemented by demanding
that $\AAA$ sends a certain subset $W$ of the $1$-ary operations in $\O_{\ovr{\CC}}$, 
the so-called Cauchy morphisms, to isomorphisms in $\TT$.
\sk

A modern research stream in AQFT is the exploration of homotopical 
phenomena associated with quantum gauge theories, which has been initiated
in \cite{BSWhomotopy} by the development of model categories that describe 
AQFTs taking values in the closed symmetric monoidal model category $\TT = \Ch_\bbK$ 
of (possibly unbounded) chain complexes of $\bbK$-vector spaces. This provides a suitable axiomatic 
framework for quantum gauge theories on globally hyperbolic Lorentzian manifolds,
which captures the explicit examples obtained from the BRST/BV formalism for AQFT
\cite{FredenhagenRejzner,FredenhagenRejzner2}
or from employing the techniques of derived geometry \cite{LinearYM,BMS}.
A new phenomenon of such homotopical AQFTs is that the time-slice axiom
gets relaxed to what we call the {\em homotopy time-slice axiom},
which demands that a $\Ch_\bbK$-valued AQFT $\AAA : \O_{\ovr{\CC}}\to \Ch_{\bbK}$
sends every Cauchy morphism in $W$ to a {\em quasi-isomorphism} of chain complexes.
This homotopical relaxation from the strict time-slice axiom (involving isomorphisms) 
to the homotopy time-slice axiom (involving quasi-isomorphisms) is not only
natural from a homotopical algebra perspective to obtain an axiom that is stable
under weak equivalences of AQFTs, but it is also the variant of the time-slice axiom
that is fulfilled by the typical examples of quantum gauge theories
constructed, for instance, in \cite{FredenhagenRejzner,FredenhagenRejzner2} and \cite{LinearYM,BMS}.
\sk

While from an abstract point of view the relaxation from the strict to the homotopy 
time-slice axiom does not cause any serious complications, it unfortunately
does have a considerable impact on the concrete applicability of $\Ch_\bbK$-valued AQFTs
to physically motivated problems. For instance, formulating the relative Cauchy evolution \cite{BFV,FewsterVerch}
or setting up local measurement schemes \cite{FewsterVerchmeasurement} in this richer homotopical
context is considerably more involved than in the case of ordinary AQFTs, because
the maps $\AAA(f) : \AAA(M)\to \AAA(N)$ associated with Cauchy morphisms $f:M\to N$
do in general not admit strict inverses and working instead with quasi-inverses 
produces a tower of homotopy coherence data that is hard to control. 
A suitable strategy to circumvent these issues is to establish
{\em strictification theorems} that allow one to replace the homotopy time-slice 
axiom by the strict one. In fact, if it would be possible to replace the $\Ch_\bbK$-valued
AQFT $\AAA$ satisfying the homotopy time-slice axiom by a weakly equivalent AQFT $\AAA_{\mathrm{st}}$
that satisfies this axiom strictly, one could avoid all the practical complications mentioned 
above by working simply with the equivalent model $\AAA_{\mathrm{st}}$ instead of $\AAA$.
A first example of such strictification theorems, which is valid
for the special case where $\ovr{\CC}$ is the RCE category
and $\AAA$ is a linear homotopy AQFT, has been proven recently in \cite{BFSRCE}.
\sk

The aim of the present paper is to prove a variety of strictification theorems
for the homotopy time-slice axiom of AQFTs, including in particular the following relevant cases:
\begin{enumerate}[label=(\roman*)]
\item $\Ch_\bbK$-valued Haag-Kastler-type AQFTs on a fixed globally 
hyperbolic Lorentzian manifold $M$ (with or without time-like boundary);

\item $\Ch_\bbK$-valued locally covariant {\em conformal} AQFTs in two spacetime dimensions;

\item $\Ch_\bbK$-valued locally covariant AQFTs in one spacetime dimension;

\item $\Ch_\bbK$-valued AQFTs on the RCE category, generalizing the result in \cite{BFSRCE}.
\end{enumerate}
Unfortunately, we currently do not know if there exists a strictification theorem
for the homotopy time-slice axiom of $\Ch_\bbK$-valued locally covariant AQFTs
in spacetime dimension $m\geq 2$.
Our strictification theorems are not only more general than the one in \cite{BFSRCE},
but they are also stronger and more powerful in the following sense:
Instead of asking
the object-wise question whether the homotopy time-slice axiom of an individual 
$\Ch_\bbK$-valued AQFT $\AAA$ can be strictified, we address the
global strictification problem that asks whether the model category
$\mathcal{L}_{\widehat{W}}\AQFT(\ovr{\CC})$ of AQFTs that satisfy
the homotopy time-slice axiom (see \cite{Carmona} and Theorem \ref{theo:Bousfield}) 
is Quillen equivalent to the model category $\AQFT(\ovr{\CC}[W^{-1}])$ of AQFTs
that satisfy the strict time-slice axiom (see Corollary \ref{cor:TimeSlice via LocalizationOfOrtCat}).
Such global strictification theorems not only provide a solution of the object-wise
strictification problem, which as mentioned above is very useful for physical applications 
such as RCE or measurement schemes, but they are also conceptually very interesting as they 
show that the homotopy time-slice axiom has no higher homotopical content in these cases. 
The availability of strictification theorems seems to be a phenomenon
that is strongly tied to AQFTs on globally hyperbolic {\em Lorentzian} manifolds
and, in particular, it is linked to the one-dimensional behavior of time evolution.
Analogous results are not available for topological QFTs, formulated in 
the context of locally constant prefactorization algebras as in 
\cite{CostelloGwilliam,CostelloGwilliam2}, where the higher homotopical content 
of $m$-dimensional isotopy equivalences gives rise to the homotopically non-trivial 
$\mathsf{E}_m$-operads, see \cite[Theorem 5.4.5.9]{HigherAlgebra}, \cite{AyalaFrancis} or \cite{CFM}.
\sk

The outline of the remainder of this paper is as follows: 
In Section \ref{sec:prelim}, we collect the relevant preliminaries for this work.
Subsection \ref{subsec:AQFT} recalls some key aspects of
orthogonal categories and their associated AQFT operads from \cite{BSWoperad}.
Subsection \ref{subsec:localization} develops systematically a theory of localizations
of orthogonal categories and explains how these can be used to implement
the strict time-slice axiom, see in particular Corollary \ref{cor:TimeSlice via LocalizationOfOrtCat}.
Subsection \ref{subsec:modelstructures} recalls the projective model structure for
$\Ch_\bbK$-valued AQFTs from \cite{BSWhomotopy} and also the left Bousfield
localized model structure from \cite{Carmona} that captures the homotopy time-slice axiom.
The Quillen adjunction established in Proposition \ref{prop:localizedadjunction} 
is the key result that allows us to formulate and prove our strictification theorems
for the homotopy time-slice axiom. Our first strictification theorem
applies to reflective localizations of orthogonal categories (in the sense of Definition 
\ref{def:OCatreflocalization}) and it is proven in Section \ref{sec:reflectivestrictification},
see in particular Theorem \ref{theo:strictificationreflective}.
This strictification theorem covers the examples (i), (ii) and (iii) from the itemization above.
In Section \ref{sec:emptystrictification}, we focus on the case in which the orthogonal category
$\ovr{\CC}$ carries an empty orthogonality relation $\perp_{\CC}^{}= \emptyset$. We then prove in Theorem
\ref{theo:emptystrictification} that in this case a sufficient condition for a strictification theorem
for the homotopy time-slice axiom of AQFTs is that a simpler strictification problem at the level
of $\Ch_\bbK$-valued functors can be solved. Building on earlier results from \cite{BFSRCE}, 
we then show that this is the case for the RCE category, which leads to example (iv) from the itemization above.
Appendix \ref{app:Loc1} proves that the localization of the category $\Loc_1$
of connected globally hyperbolic Lorentzian $1$-manifolds at all Cauchy morphisms is reflective, which we need
to deduce example (iii) from Theorem \ref{theo:strictificationreflective}.


\section{\label{sec:prelim}Preliminaries}

\subsection{\label{subsec:AQFT}Orthogonal categories, colored operads and AQFTs}
Orthogonal categories \cite{BSWoperad} are an abstraction of the 
concept of a category of spacetimes with a notion of causally independent pairs of subregions  
$f_1 : M_1 \to N \leftarrow M_2 : f_2$.
The relevant definitions are as follows:
\begin{defi}
\begin{itemize}
\item[a)] An {\em orthogonal category} is a pair $\ovr{\CC} := (\CC,\perp_{\CC}^{})$ consisting of
a small category $\CC$ and a subset $\perp_{\CC}^{}  \, \subseteq \Mor\,\CC \,{}_{\t}{\times}_{\t} \,\Mor\,\CC$
(called orthogonality relation) 
of the set of pairs of morphisms to a common target, such that the following conditions hold true:
\begin{itemize}
\item[(i)] Symmetry: $(f_1,f_2)\in\,\perp_{\CC}^{}$ implies $(f_2,f_1)\in\,\perp_{\CC}^{}$.
\item[(ii)] Composition stability: $(f_1,f_2)\in \,\perp_{\CC}^{}$ implies
$(g\,f_1\,h_1, g\,f_2\,h_2)\in \,\perp_{\CC}^{}$
for all composable morphisms $g$, $h_1$ and $h_2$.
\end{itemize}
We often write $f_1\perp_{\CC}^{} f_2$ instead of $(f_1,f_2)\in\, \perp_{\CC}^{}$ to denote 
orthogonal pairs of morphisms.

\item[b)] An {\em orthogonal functor} $F : \ovr{\CC}\to\ovr{\DD}$ is a functor $F : \CC\to \DD$
between the underlying categories that preserves orthogonal pairs, i.e.\ 
$F(f_1)\perp_{\DD}^{} F(f_2)$ for all $f_1\perp_{\CC}^{} f_2$.

\item[c)] We denote by $\OCat$ the $2$-category whose objects are
orthogonal categories, $1$-morphisms are orthogonal functors and $2$-morphisms
are natural transformations between orthogonal functors.
\end{itemize}
\end{defi}

We would like to note in passing that the 
equivalences in the $2$-category $\OCat$ can be characterized explicitly.
\begin{lem}[{\cite[Lemma 2.7]{Skeletal}}]
An orthogonal functor $F : \ovr{\CC}\to \ovr{\DD}$ is an equivalence
in the $2$-category $\OCat$ if and only if the following two conditions hold true:
\begin{itemize}
\item[(i)] The underlying functor $F :\CC\to\DD$ is fully faithful and essentially surjective.
\item[(ii)] The orthogonality relation $\perp_{\CC}^{}\, = F^\ast(\perp_{\DD}^{})$ 
agrees with the pullback along $F$ (see \cite[Lemma 3.19]{BSWoperad}) of the orthogonality relation $\perp_\DD^{}$,
i.e.\ $f_1\perp_{\CC}^{}f_2$ if and only if $F(f_1)\perp_{\DD}^{}F(f_2)$.
\end{itemize}
\end{lem}

It was shown in \cite{BSWoperad} that, associated to each orthogonal category
$\ovr{\CC}$, there is a colored operad $\O_{\ovr{\CC}}^{}$ 
that codifies the algebraic structure of AQFTs on $\ovr{\CC}$.
See also \cite{BSreview} for a concise review.
Recall that a colored operad (aka multicategory) is a generalization
of the concept of a category in which morphisms may have multiple inputs. 
More precisely, a colored operad $\P$ consists of
the following data:
\begin{itemize}
\item[(i)] a class of objects, sometimes also called colors;
\item[(ii)] for each tuple $(\und{c},t) = ((c_1,\dots,c_n),t)$ of objects, 
a set of operations $\P\big(\substack{t \\ \und{c}}\big)$ from $\und{c}$ to $t$;
\item[(iii)] composition maps $\gamma : \P\big(\substack{t \\ \und{c}}\big)\times \prod_{i=1}^n\P\big(\substack{c_i \\ \und{d}_i}\big) \to \P\big(\substack{t \\ (\und{d}_1,\dots,\und{d}_n)}\big)$;
\item[(iv)] unit elements $\oone \in \P\big(\substack{t\\ t}\big)$;
\item[(v)] permutation actions $\P(\sigma) : \P\big(\substack{ t \\ \und{c}}\big)\to \P\big(\substack{t \\ \und{c}\sigma}\big)$, for each $\sigma\in \Sigma_n$, where $\und{c}\sigma = (c_{\sigma(1)},\dots,c_{\sigma(n)})$.
\end{itemize}
These data have to satisfy the usual associativity, unitality and equivariance axioms, 
see e.g.\ \cite{Yau} for the details. Similarly to categories, we often denote
an operation $\phi \in \P\big(\substack{t \\ \und{c}}\big)$ by an arrow $\phi : \und{c}\to t$.
Given also a tuple of operations $\und{\psi} = (\psi_1,\dots,\psi_n)$, with $\psi_i : \und{d}_i\to c_i$,
we denote the composition by juxtaposition $\phi\,\und{\psi}  := 
\gamma(\phi, (\psi_1,\dots,\psi_n)) : \und{\und{d}}\to t$, where $\und{\und{d}} = (\und{d}_1,\dots,\und{d}_n)$.
The permutation actions will be denoted by dots $\phi\cdot \sigma := \P(\sigma)(\phi) : \und{c}\sigma\to t$.
Colored operads assemble into a $2$-category, see e.g.\ \cite{Mandell,Weiss} for the details.
\begin{defi}
We denote by $\Op$ the $2$-category whose objects are colored operads, $1$-morphisms 
are multifunctors and $2$-morphisms are multinatural transformations.
\end{defi}

The AQFT operad $\O_{\ovr{\CC}}^{}$ admits the following explicit description, see \cite{BSWoperad,BSreview}. 
\begin{defi}\label{def:AQFToperad}
Let $\ovr{\CC} = (\CC,\perp_{\CC}^{}) \in \OCat$ be an orthogonal category.
The associated {\em AQFT operad} $\O_{\ovr{\CC}}^{}\in \Op$ is the colored operad
that is defined by the following data:
\begin{itemize}
\item[(i)] the objects are the objects of $\CC$;

\item[(ii)] the set of operations from $\und{M} = (M_1,\dots, M_n)$
to $N$ is the quotient set
\begin{flalign}
\O_{\ovr{\CC}}^{}\big(\substack{N \\ \und{M}}\big)\,:=\,\bigg(\Sigma_n \times \prod_{i=1}^n \CC(M_i,N)\bigg)\Big/\!\sim_{\perp_{\CC}^{}}^{}\quad,
\end{flalign}
where $\CC(M_i,N)$ denotes the set of $\CC$-morphisms from $M_i$ to $N$
and the equivalence relation is defined as follows: $(\sigma,\und{f}) 
\sim_{\perp_{\CC}^{}}^{} (\sigma^\prime,\und{f}^\prime)$ if and only if 
$\und{f} = \und{f}^\prime$ and the right permutation
$\sigma\sigma^{\prime -1} : \und{f}\sigma^{-1}\to \und{f}\sigma^{\prime -1}$
is generated by transpositions of adjacent orthogonal pairs;

\item[(iii)] the composition of $[\sigma,\und{f}] : \und{M}\to N$ 
with $[\sigma_i,\und{g}_i] : \und{K}_i\to M_i$, for $i=1,\dots,n$, is
\begin{subequations}
\begin{flalign}
[\sigma , \und{f}]\, [ \und{\sigma}, \und{\und{g}}]\,:=\,
\big[\sigma(\sigma_1,\dots,\sigma_n), \und{f}\,\und{\und{g}}\,\big]\,:\,\und{\und{K}}\,\longrightarrow\,N\quad,
\end{flalign}
where $\sigma(\sigma_1,\dots,\sigma_n)$ denotes
the composition in the unital associative operad and
\begin{flalign}
 \und{f}\,\und{\und{g}} \,:=\, \big(f_1\,g_{11},\dots, f_1\,g_{1k_1},\dots, f_n\,g_{n1},\dots,f_{n}\,g_{n k_n}\big)
\end{flalign}
\end{subequations}
is given by compositions in the category $\CC$;

\item[(iv)] the unit elements are $\oone := [e,\id_N^{}] : N\to N$, where $e\in\Sigma_1$ is the identity permutation;

\item[(v)] the permutation action of $\sigma^\prime\in\Sigma_n$ on $[\sigma,\und{f}] : \und{M}\to N$ is
\begin{flalign}
[\sigma,\und{f}]\cdot \sigma^\prime \,:=\, [\sigma\sigma^\prime,\und{f}\sigma^\prime] : \und{M}\sigma^{\prime}~\longrightarrow~
 N\quad,
\end{flalign}
where $\und{f}\sigma^\prime = (f_{\sigma^\prime(1)},\dots,f_{\sigma^\prime(n)})$ 
and $\und{M}\sigma^{\prime}= (M_{\sigma^\prime(1)},\dots, M_{\sigma^{\prime}(n)})$ denote 
the permuted tuples and $\sigma\sigma^\prime$ is given by the group operation of the permutation group $\Sigma_n$.
\end{itemize}
\end{defi}
\begin{rem}\label{rem:genrel}
There exists a useful presentation of the colored operad $\O_{\ovr{\CC}}^{}$ in terms of generators
and relations, see \cite[Section 3.3]{BSWoperad} and also \cite[Section 2.2]{BSreview} for a concise review.
Very briefly, there are three types of generators
\begin{flalign}
\text{\begin{tabular}{c}
\begin{tikzpicture}[cir/.style={circle,draw=black,inner sep=0pt,minimum size=2mm},
        poin/.style={circle, inner sep=0pt,minimum size=0mm},scale=0.8, every node/.style={scale=0.8}]
%
%
\node[poin] (Hout) [label=above:{\small $N$}] at (0,1) {};
\node[poin] (Hin) [label=below:{\small $M$}] at (0,0) {};
\draw[thick] (Hin) -- (Hout) node[midway,left] {{\small $f$}};
%
%
\node[poin] (Uout) [label=above:{\small $N$}] at (2,1) {};
\node[poin] (Uin) [label=below:{\small $\emptyset$}] at (2,0) {};
\draw[thick] (Uin) -- (Uout) node[midway,left] {{\small $1_N^{}$}};
\draw[thick,fill=white] (Uin) circle (0.6mm);
%
%
\node[poin] (Mout) [label=above:{\small $N$}] at (4,1) {};
\node[poin] (Min1) [label=below:{\small $N$}] at (3.7,0) {};
\node[poin] (Min2) [label=below:{\small $N$}] at (4.3,0) {};
\node[poin] (V)  [label=right:{\small $\mu_N^{}$}] at (4,0.5) {};
\draw[thick] (Min1) -- (V);
\draw[thick] (Min2) -- (V);
\draw[thick] (V) -- (Mout);
\end{tikzpicture}
\end{tabular}}
\end{flalign}
that, in AQFT terminology, describe the pushforward of observables along morphisms 
$f:M\to N$ in $\CC$, the unit observable on $N$, and the multiplication of observables on $N$.
These generators have to satisfy various relations, which can be classified into
functoriality relations, algebra relations, compatibility relations and
$\perp_{\CC}^{}$-commutativity relations. The latter may be visualized as
\begin{flalign}
\text{\begin{tabular}{c}
\begin{tikzpicture}[cir/.style={circle,draw=black,inner sep=0pt,minimum size=2mm},
        poin/.style={circle, inner sep=0pt,minimum size=0mm},scale=0.8, every node/.style={scale=0.8}]
\node[poin] (Mout) [label=above:{\small $N$}] at (0,1.25) {};
\node[poin] (Min1) [label=below:{\small $M_1$}] at (-0.5,-0.25) {};
\node[poin] (Min2) [label=below:{\small $M_2$}] at (0.5,-0.25) {};
\node[poin] (V)  [label=right:{\small $\mu_N^{}$}] at (0,0.5) {};
\draw[thick] (Min1) -- (V) node[midway,left] {{\small $f_1$}};
\draw[thick] (Min2) -- (V) node[midway,right] {{\small $f_2$}};
\draw[thick] (V) -- (Mout);
\node[poin] (Eq) at (1.25,0.5) {$~~=~~$};
\node[poin] (MMout) [label=above:{\small $N$}] at (2.5,1.5) {};
\node[poin] (MMin1)  at (2,0) {};
\node[poin] (MMin2)  at (3,0) {};
\node[poin] (MMin11) [label=below:{\small $M_1$}] at (2,-0.5) {};
\node[poin] (MMin21) [label=below:{\small $M_2$}] at (3,-0.5) {};
\node[poin] (VV)  [label=right:{\small $\mu_N^{}$}] at (2.5,0.75) {};
\draw[thick] (MMin1) -- (VV) node[midway,left] {{\small $f_2$}};
\draw[thick] (MMin2) -- (VV) node[midway,right] {{\small $f_1$}};
\draw[thick] (VV) -- (MMout);
\draw[thick] (MMin11) -- (MMin2);
\draw[thick] (MMin21) -- (MMin1);
\end{tikzpicture}
\end{tabular}}\quad,
\end{flalign}
for all $(f_1:M_1\to N)\perp_{\CC}^{} (f_2:M_2\to N)$,
and their role is to implement the quotient in 
the set of operations from Definition \ref{def:AQFToperad}.
The physical interpretation of the $\perp_{\CC}^{}$-commutativity relations
is that they enforce `Einstein causality' for all orthogonal pairs of morphisms.
\end{rem}

The assignment $\ovr{\CC}\mapsto \O_{\ovr{\CC}}^{}$ of AQFT operads can be 
upgraded to a $2$-functor. Given an orthogonal functor $F : \ovr{\CC}\to \ovr{\DD}$,
we define the multifunctor
\begin{flalign}
\nn \O_F^{}\,:\, \O_{\ovr{\CC}}^{}~&\longrightarrow~\O_{\ovr{\DD}}^{}\quad,\\
\nn M~&\longmapsto~ F(M)\quad,\\
\big([\sigma,\und{f}]:\und{M}\to N\big)~&\longmapsto ~\big([\sigma,F(\und{f})] : F(\und{M})\to F(N)\big)\quad,\label{eqn:OFformula}
\end{flalign}
where $F$ is defined to act component-wise on tuples, i.e.\ $F(\und{M}) = (F(M_1),\dots,F(M_n))$
and also $F(\und{f}) = (F(f_1),\dots,F(f_n))$. Given further a natural transformation
$\chi: F\to G$ between orthogonal functors $F,G : \ovr{\CC}\to\ovr{\DD}$,
we define the multinatural transformation $\O_{\chi}^{} : \O_{F}^{}\to \O_{G}^{}$ between
the corresponding multifunctors $\O_{F}^{},\O_{G}^{} : \O_{\ovr{\CC}}^{}\to \O_{\ovr{\DD}}^{}$
by setting
\begin{flalign}
(\O_{\chi})_{M}^{}\,:=\,[e,\chi_M^{}]\, :\, F(M)~\longrightarrow~ G(M) \quad,
\end{flalign}
for all components $M\in \CC$. Summing up, we obtain
\begin{propo}\label{propo:Operad2functor}
The above defines a $2$-functor
\begin{flalign}\label{eqn:Operad2functor}
\O \,:\, \OCat~\longrightarrow~\Op
\end{flalign}
from the $2$-category of orthogonal categories to the $2$-category of colored operads.
\end{propo}

AQFTs on $\ovr{\CC}$ are by definition algebras over the AQFT operad $\O_{\ovr{\CC}}^{}$
with values in a suitable closed symmetric monoidal category $\TT$. For ordinary AQFTs,
one usually takes $\TT=\Vec_{\bbK}$ to be the category of vector spaces
over a field $\bbK$ of characteristic $0$, while for homotopy AQFTs one takes $\TT=\Ch_{\bbK}$
to be the category of chain complexes. (The relevant model categorical aspects in the latter
case will be discussed later in Subsection \ref{subsec:modelstructures}.)
From our $2$-categorical perspective, there is the following slick definition
of the category of AQFTs on $\ovr{\CC}$. Recall that to each 
symmetric monoidal category $\TT$ one can assign a colored operad (which we denote with abuse of notation 
by the same symbol) that has the same objects and whose sets of operations are given by
\begin{flalign}
\TT\big(\substack{Y\\ \und{X}}\big)\,:=\, \TT\bigg(\Motimes_{i=1}^n X_i , Y\bigg)\quad,
\end{flalign}
for all tuples $(\und{X},Y)= ((X_1,\dots,X_n),Y)$ of objects in $\TT$.
Operadic composition is defined by compositions and tensor products of $\TT$-morphisms, the unit elements
correspond to the identity morphisms $\oone = \id_Y \in \TT\big(\substack{Y\\ Y}\big) = \TT(Y,Y)$
and the permutation actions are defined via the symmetric braiding.
\begin{defi}\label{def:AQFT}
Let $\ovr{\CC}\in \OCat$ be an orthogonal category and $\TT$ a closed symmetric monoidal category.
The {\em category of $\TT$-valued AQFTs on $\ovr{\CC}$} is defined
as the $\Hom$-category
\begin{flalign}
\AQFT(\ovr{\CC})\,:=\, \Alg_{\O_{\ovr{\CC}}^{}}^{}\big(\TT\big) \,:=\, \Hom_{\Op}^{}\big(\O_{\ovr{\CC}}^{},\TT\big)
\end{flalign}
in the $2$-category $\Op$ from the AQFT operad $\O_{\ovr{\CC}}$
to the colored operad associated to $\TT$. More explicitly,
an object in this category is a multifunctor $\AAA : \O_{\ovr{\CC}}^{}\to \TT$
and a morphism between two objects $\AAA,\BBB : \O_{\ovr{\CC}}^{}\to \TT$
is a multinatural transformation $\zeta : \AAA\to \BBB$.
\end{defi}

Observe that the assignment $\ovr{\CC}\mapsto \AQFT(\ovr{\CC})$
admits a canonical upgrade to a $2$-functor
\begin{flalign}\label{eqn:AQFT2functor}
\AQFT\,:\, (\OCat)^{\op}~\longrightarrow~\Cat
\end{flalign} 
to the $2$-category $\Cat$ of (not necessarily small) categories, functors and natural transformations.
Indeed, for each orthogonal functor $F : \ovr{\CC}\to\ovr{\DD}$,
we can define a pullback functor
\begin{flalign}\label{eqn:AQFTpullback}
F^\ast \,:=\,(-)\O_F^{}\,:\, \AQFT(\ovr{\DD})~\longrightarrow~\AQFT(\ovr{\CC})
\end{flalign}
that acts on objects $(\AAA:\O_{\ovr{\DD}}^{}\to \TT) 
\mapsto (\AAA\,\O_{F}^{} : \O_{\ovr{\CC}}^{}\to \TT)$ by pre-composition
and on morphisms $\zeta \mapsto \zeta\,\O_{F}^{}$ by whiskering in the $2$-category $\Op$.
The action of the $2$-functor \eqref{eqn:AQFT2functor} on $2$-morphisms
is given by whiskering too, i.e.\ for each natural transformation
$\chi : F\to G$ between orthogonal functors we define
$\chi^\ast := (-)\O_{\chi}^{} : (-)\O_{F}^{}\to (-)\O_{G}^{}$.
\sk

We conclude this subsection by recording the following standard
result about operadic left Kan extensions, see e.g.
\cite[Theorem 2.11]{BSWoperad} for a spelled out proof.
\begin{propo}\label{prop:LKE}
Suppose that the closed symmetric monoidal category $\TT$ is cocomplete, i.e.\ it admits all small 
colimits. Then, for each orthogonal functor $F : \ovr{\CC}\to\ovr{\DD}$,
the pullback functor in \eqref{eqn:AQFTpullback} admits a left adjoint, i.e.\ we obtain
an adjunction
\begin{flalign}
\xymatrix{
F_! \,:\, \AQFT(\ovr{\CC})\ar@<0.5ex>[r]~&~\ar@<0.5ex>[l]\AQFT(\ovr{\DD})\,:\,F^\ast
}\quad.
\end{flalign}
\end{propo}


\subsection{\label{subsec:localization}Time-slice axiom and localizations}
Note that our Definition \ref{def:AQFT} of AQFTs does {\em not} explicitly
refer to one of the central physical axioms, 
namely the time-slice axiom. The goal of this subsection
is to clarify, in the case where the target $\TT$ is an 
ordinary category (in contrast to a model category such as $\TT=\Ch_\bbK$),
the precise sense in which the time-slice axiom may be implemented
by a localization of orthogonal categories (see Corollary \ref{cor:TimeSlice via LocalizationOfOrtCat}). 
\sk

Let us start by making precise the concept of 
localization of orthogonal categories, which has previously appeared
in a more {\em ad hoc} fashion in \cite[Section 4.2]{BSWoperad}.
Similarly to ordinary category theory, localizations in $\OCat$
are characterized by the following universal property.
\begin{defi}\label{def:OCatlocalization}
Let $\ovr{\CC}$ be an orthogonal category
and $W\subseteq \Mor\,\CC$ a subset. A {\em localization} of $\ovr{\CC}$ at
$W$ is an orthogonal category $\ovr{\CC}[W^{-1}]$ together with
an orthogonal functor $L:\ovr{\CC}\to\ovr{\CC}[W^{-1}]$ satisfying the following properties:
\begin{itemize}
\item[(i)] For all morphisms $(f:M\to N) \in W$, the morphism
$L(f) : L(M)\to L(N)$ is an isomorphism in  $\ovr{\CC}[W^{-1}]$.

\item[(ii)] For any orthogonal category $\ovr{\DD}$ and any orthogonal functor
$F : \ovr{\CC}\to\ovr{\DD}$ that sends morphisms in $W$ to isomorphisms in $\ovr{\DD}$,
there exists an orthogonal functor $F_W^{} : \ovr{\CC}[W^{-1}]\to \ovr{\DD}$
and a natural isomorphism $F\cong F_W^{}\,L$.

\item[(iii)] For all orthogonal functors $G,H:\ovr{\CC}[W^{-1}]\to \ovr{\DD}$,
the whiskering map
\begin{flalign}
(-)L\,:\,\mathrm{Nat}(G,H)~\longrightarrow~\mathrm{Nat}(G L,H L)
\end{flalign}
between the sets of $2$-morphisms (i.e.\ natural transformations) is a bijection.
\end{itemize}
\end{defi}
\begin{rem}\label{rem:OCatlocalization}
Denoting by $\Hom_{\OCat}^{}\big(\ovr{\CC},\ovr{\DD}\big)\in\Cat$ the $\Hom$-category between
two objects $\ovr{\CC}$ and $\ovr{\DD}$ in the $2$-category $\OCat$ of orthogonal categories, 
the properties from Definition \ref{def:OCatlocalization} can be rephrased 
more concisely as follows: For every orthogonal category $\ovr{\DD}$, the functor
\begin{subequations}
\begin{flalign}
(-)L\,:\, \Hom_{\OCat}^{}\big(\ovr{\CC}[W^{-1}],\ovr{\DD}\big)~\longrightarrow~\Hom_{\OCat}^{}\big(\ovr{\CC},\ovr{\DD}\big)^W
\end{flalign}
is an equivalence of categories, where 
\begin{flalign}
\Hom_{\OCat}^{}\big(\ovr{\CC},\ovr{\DD}\big)^W\,\subseteq \,
\Hom_{\OCat}^{}\big(\ovr{\CC},\ovr{\DD}\big)
\end{flalign}
\end{subequations}
denotes the full subcategory of orthogonal functors that send $W$ to isomorphisms.
\end{rem}

It is easy to prove that the {\em ad hoc} concept 
of localization of orthogonal categories from \cite[Section 4.2]{BSWoperad}
defines a localization in the sense of Definition \ref{def:OCatlocalization}.
\begin{propo}\label{prop:Olocalizationmodel}
Let $\ovr{\CC} = (\CC,\perp_{\CC}^{})$ be an orthogonal category
and $W\subseteq \Mor\,\CC$ a subset.
Consider the localization $L : \CC\to \CC[W^{-1}]$
of the underlying category $\CC$ at $W$ and endow $\CC[W^{-1}]$ with the pushforward
orthogonality relation $\perp_{\CC[W^{-1}]} \,:= L_\ast(\perp_{\CC}^{})$ (see \cite[Lemma 3.19]{BSWoperad}),
i.e.\ the minimal orthogonality relation such that $L : \ovr{\CC}\to \ovr{\CC[W^{-1}]}$ is 
an orthogonal functor. Then the resulting orthogonal functor $L : \ovr{\CC}\to \ovr{\CC[W^{-1}]}$
is a localization in the sense of Definition \ref{def:OCatlocalization}.
\end{propo}
\begin{proof}
Items (i) and (iii) hold true because the underlying functor $L : \CC\to\CC[W^{-1}]$ is a localization of categories.
For item (ii), we obtain at the level of the underlying categories
a functor $F_W^{} : \CC[W^{-1}]\to \DD$ and a natural isomorphism $F\cong F_W^{}\,L$.
It remains to show that $F_W^{}$ is orthogonal with respect to 
$\perp_{\CC[W^{-1}]}^{} \,= L_\ast(\perp_{\CC}^{})$. By composition stability of $\perp_{\DD}^{}$ and 
the definition of the pushforward orthogonality relation (see \cite[Lemma 3.19]{BSWoperad}), this
is the case if and only if $F_W^{} L(f_1) \perp_{\DD}^{} F_W^{}L(f_2)$, for all $f_1\perp_{\CC}^{}f_2$. 
The latter holds true because $F\cong F_W^{}\,L$ are naturally isomorphic
and $\perp_{\DD}^{}$ is composition stable. Let us spell out this last 
step in more detail: Writing $(f_1 : M_1\to N)\perp_{\CC}^{}(f_2 : M_2\to N)$ and $\chi : F \to F_W^{}\,L$
for the given natural isomorphism, we obtain that
\begin{flalign}
\big(F_W^{} L(f_1),F_W^{} L(f_2)\big)\,=\, \big(\chi_{N}^{}\,F(f_1)\,\chi_{M_1}^{-1} , \chi_{N}^{}\, F(f_2)\,\chi_{M_2}^{-1}\big)\,\in\,\perp_{\DD}^{}
\end{flalign}
because $\perp_{\DD}^{}$ is composition stable and $F(f_1)\perp_{\DD}^{}F(f_2)$ 
as a consequence of $F$ being an orthogonal functor.
\end{proof}

We shall now prove that the $2$-functor $\O: \OCat\to \Op$
from Proposition \ref{propo:Operad2functor} that assigns to an orthogonal category
$\ovr{\CC}$ its AQFT operad $\O_{\ovr{\CC}}^{}$ maps localizations of orthogonal
categories to localizations of operads. The latter are characterized by
the following universal property.
\begin{defi}\label{def:operadlocalization}
Let $\P$ be a colored operad and $W\subseteq \Mor_1^{}\P$ a subset of $1$-ary operations in $\P$.
A {\em localization} of $\P$ at $W$ is a colored operad $\P[W^{-1}]$ together with a multifunctor
$L : \P\to \P[W^{-1}]$ satisfying the following properties:
\begin{itemize}
\item[(i)] For all operations $(\phi : c\to t)\in W$, the $1$-ary operation
$L(\phi) : L(c)\to L(t)$ is an isomorphism in $\P[W^{-1}]$.

\item[(ii)] For any colored operad $\Q$ and any multifunctor $F : \P\to \Q$ that sends operations
in $W$ to isomorphisms in $\Q$, there exists a multifunctor $F_W^{} : \P[W^{-1}]\to \Q$ and a multinatural
isomorphism $F\cong F_W^{}\,L$.

\item[(iii)] For all multifunctors $G,H: \P[W^{-1}]\to \Q$, the whiskering map
\begin{flalign}
(-)L\,:\,\mathrm{mNat}(G,H)~\longrightarrow~\mathrm{mNat}(G L,H L)
\end{flalign}
between the sets of $2$-morphisms (i.e.\ multinatural transformations) is a bijection.
\end{itemize}
\end{defi}
\begin{rem}\label{rem:Oplocalization}
Similarly to Remark \ref{rem:OCatlocalization}, the properties from Definition
\ref{def:operadlocalization} can be rephrased more concisely as follows:
For every colored operad $\Q$, the functor
\begin{subequations}
\begin{flalign}
(-)L\,:\, \Hom_{\Op}^{}\big(\P[W^{-1}],\Q\big)~\longrightarrow~\Hom_{\Op}^{}\big(\P,\Q\big)^W
\end{flalign}
is an equivalence of categories, where 
\begin{flalign}
\Hom_{\Op}^{}\big(\P,\Q\big)^W\, \subseteq\, \Hom_{\Op}^{}\big(\P,\Q\big)
\end{flalign} 
\end{subequations}
denotes the full subcategory of multifunctors 
that send $W$ to isomorphisms.
\end{rem}

\begin{propo}\label{prop:Operad2functorlocalization}
Suppose that $L : \ovr{\CC}\to\ovr{\CC}[W^{-1}]$ is a localization of orthogonal categories.
Applying the $2$-functor from Proposition \ref{propo:Operad2functor} defines
a multifunctor $\O_{L}^{} : \O_{\ovr{\CC}}^{} \to \O_{\ovr{\CC}[W^{-1}]}^{}$ that exhibits a
localization of the colored operad $\O_{\ovr{\CC}}^{}$ at the set of $1$-ary operations
$W\subseteq \Mor_1^{} \O_{\ovr{\CC}}^{}$.
\end{propo}
\begin{proof}
Because of $2$-functoriality and uniqueness (up to equivalence) of localizations, we can pick 
without loss of generality a particular model for the underlying localized category $\CC[W^{-1}]$.
To allow for an effective use of the generators-relations description of the AQFT operads
from Remark \ref{rem:genrel}, we work with the Gabriel-Zisman model \cite{GabrielZisman}
in which $\CC[W^{-1}]$ has the same objects as $\CC$ and its morphisms are given by 
equivalence classes of chains of zig-zags of $\CC$-morphisms, where all reverse-pointing morphisms must be in $W$.
In short, the category $\CC[W^{-1}]$ is generated by all $\CC$-morphisms $f: M\to N$
and formal inverses $w^{-1} : M \to N$ of all $W$-morphisms $w : N\to M$ modulo
the equivalence relation described in \cite{GabrielZisman}. The localization functor
$L : \CC\to \CC[W^{-1}]$ then acts on objects and morphisms simply as $M\mapsto M$ and
$(f:M\to N)\mapsto (f:M\to N)$.
\sk

We shall now explicitly verify that $\O_{L}^{}: \O_{\ovr{\CC}}^{} \to \O_{\ovr{\CC}[W^{-1}]}^{}$
satisfies the three properties from Definition \ref{def:operadlocalization} that characterize a localization of colored operads.
Item (i) is obvious from the definition of $\O_{L}^{}$ in \eqref{eqn:OFformula}. 
For item (ii), consider any multifunctor $F : \O_{\ovr{\CC}}^{} \to \Q$ 
that sends operations in $W$ to isomorphisms. We have to provide an extension of $F$ along $\O_{L}^{}$ 
to a multifunctor $F_W^{} : \O_{\ovr{\CC}[W^{-1}]}^{}\to \Q$. Using the Gabriel-Zisman model
for $\CC[W^{-1}]$ and our generators-relations description from Remark \ref{rem:genrel}, we can construct a strict
extension (i.e.\ the multinatural isomorphism $F \cong F_W^{}\,\O_{L}^{}$ is the identity) by mapping
the formal inverses of $W$-morphisms as follows
\begin{flalign}
\big(w^{-1} : M\to N\big)~\longmapsto~ \big( F(w)^{-1} : F(M)\to F(N)\big)\quad.
\end{flalign}
It is easy to check that this is compatible with the relations of $\O_{\ovr{\CC}[W^{-1}]}^{}$, hence we
obtain a multifunctor $F_W^{} : \O_{\ovr{\CC}[W^{-1}]}^{}\to \Q$ such that $F=F_W^{}\,\O_{L}^{}$.
\sk

For item (iii), injectivity is obvious and surjectivity is shown by the 
following argument: Given any multinatural transformation 
$\zeta : G\,\O_{L}^{}\to H\,\O_{L}^{}$ of multifunctors from $\O_{\ovr{\CC}}^{}$
to $\Q$, we have to prove that the underlying components
$\big\{\zeta_{M}^{} : G(M)\to H(M)\big\}_{M\in\CC}$ are 
multinatural on $\O_{\ovr{\CC}[W^{-1}]}^{}$ too. Since multinaturality
can be checked at the level of the generators, this amounts to showing
that for each formal inverse $w^{-1} : M\to N$ the naturality diagram
\begin{flalign}
\xymatrix@C=3em{
\ar[d]_-{\zeta_M^{}}G(M) \ar[r]^-{G(w^{-1})} & G(N)\ar[d]^-{\zeta_N^{}}\\
H(M) \ar[r]_-{H(w^{-1})} & H(N)
}
\end{flalign}
in $\Q$ commutes. Since $G(w^{-1})=G(w)^{-1}$ and $H(w^{-1}) = H(w)^{-1}$
in $\Q$, these naturality conditions are already enforced by the arity $1$ generators of $\O_{\ovr{\CC}}^{}$.
\end{proof}

The relevance of this result for AQFT can be summarized as follows.
\begin{cor}\label{cor:TimeSlice via LocalizationOfOrtCat}
Let $\ovr{\CC}$ be an orthogonal category and $W\subseteq \Mor\,\CC$ a subset. Denote by
\begin{flalign}
\AQFT(\ovr{\CC})^W\,\subseteq \, \AQFT(\ovr{\CC})
\end{flalign}
the full subcategory of AQFTs on $\ovr{\CC}$ that send $W$ to isomorphisms in $\TT$.
In AQFT terminology, one may say that such theories satisfy the {\em time-slice axiom}
for the set of morphisms $W$.
Then the pullback functor
\begin{flalign}
L^\ast\,:=\,(-)\O_{L}^{}\,:\, \AQFT\big(\ovr{\CC}[W^{-1}]\big)~\longrightarrow~\AQFT(\ovr{\CC})^W
\end{flalign}
associated to the localization $L : \ovr{\CC}\to \ovr{\CC}[W^{-1}]$ of orthogonal categories
defines an equivalence between the category of AQFTs on $\ovr{\CC}[W^{-1}]$ and the category
of AQFTs on $\ovr{\CC}$ that satisfy the time-slice axiom for $W$.
\end{cor}
\begin{proof}
By Proposition \ref{prop:Operad2functorlocalization}, the multifunctor 
$\O_{L}^{} : \O_{\ovr{\CC}}^{}\to\O_{\ovr{\CC}[W^{-1}]}$ is a localization of the colored operad
$\O_{\ovr{\CC}}^{}$ at $W$, hence the statement is a direct consequence of the universal property
from Remark \ref{rem:Oplocalization}.
\end{proof}
\begin{rem}
For a cocomplete closed symmetric monoidal category $\TT$, 
Corollary \ref{cor:TimeSlice via LocalizationOfOrtCat} can be rephrased in the language of
Proposition \ref{prop:LKE}. This will be useful to 
understand our strategy and constructions in Sections 
\ref{sec:reflectivestrictification} and \ref{sec:emptystrictification}.
The orthogonal localization functor $L : \ovr{\CC}\to \ovr{\CC}[W^{-1}]$
defines by Proposition \ref{prop:LKE} an adjunction
\begin{flalign}\label{eqn:Ladjunction}
\xymatrix{
L_! \,:\, \AQFT(\ovr{\CC})\ar@<0.5ex>[r]~&~\ar@<0.5ex>[l]\AQFT\big(\ovr{\CC}[W^{-1}]\big)\,:\,L^\ast
}\quad,
\end{flalign}
whose right adjoint is the pullback functor $ L^\ast =(-)\O_{L}^{}$ and whose
left adjoint $L_!$ is given by operadic left Kan extension
(see e.g.\ \cite[Proposition 2.12]{BSWoperad} for an explicit colimit formula).
Using this adjunction, we can associate to any $\AAA\in \AQFT(\ovr{\CC})$
the AQFT $L^\ast L_!(\AAA)\in\AQFT(\ovr{\CC})^W $ that satisfies the time-slice axiom,
whether or not the original $\AAA$ satisfies this axiom,
together with a comparison morphism $\eta_{\AAA}^{} : \AAA\to L^\ast L_!(\AAA)$ 
given by the adjunction unit $\eta: \id\to L^\ast\,L_!$. In general, the comparison morphism
will \textit{not} be an isomorphism, hence $\AAA$ and $L^\ast L_!(\AAA)$ define different AQFTs.
The result in Corollary \ref{cor:TimeSlice via LocalizationOfOrtCat} implies
that the adjunction \eqref{eqn:Ladjunction} induces an adjoint equivalence
\begin{flalign}
\xymatrix{
L_! \,:\, \AQFT\big(\ovr{\CC})^W\ar@<0.75ex>[r]_-{\sim}~&~\ar@<0.75ex>[l]\AQFT\big(\ovr{\CC}[W^{-1}]\big)\,:\,L^\ast
}
\end{flalign}
when we restrict to the full subcategory $\AQFT(\ovr{\CC})^W\subseteq \AQFT(\ovr{\CC})$
of AQFTs that satisfy the time-slice axiom. As a consequence,
the comparison morphism $\eta_{\AAA}^{} : \AAA\to L^\ast L_!(\AAA)$ is an isomorphism
if and only if the theory $\AAA\in\AQFT(\ovr{\CC})^W$  satisfies the time-slice axiom.
\end{rem}

\subsection{\label{subsec:modelstructures}Model category structures}
For the rest of this paper, we fix the target category
\begin{flalign}
\TT\,:=\, \Ch_\bbK
\end{flalign}
to be the closed symmetric monoidal category of (possibly unbounded) chain 
complexes of vector spaces over a field $\bbK$ of characteristic $0$. 
By \cite[Sections 2.3 and 4.2]{Hovey}, $\TT=\Ch_\bbK$  carries the structure
of a {\em closed symmetric monoidal model category} in which the weak equivalences
are the quasi-isomorphisms and the fibrations are the degree-wise surjective chain maps.
As a consequence of Hinich's results \cite{Hinich1,Hinich2}, this model structure can be transferred
to the category of $\Ch_{\bbK}$-valued AQFTs on an orthogonal category $\ovr{\CC}$.
The following model categories appeared first in \cite{BSWhomotopy}.
\begin{theo}\label{theo:projmodel}
For each orthogonal category $\ovr{\CC}$, the associated category
\begin{flalign}
\AQFT(\ovr{\CC})\,:=\,\Alg_{\O_{\ovr{\CC}}^{}}^{}\big(\Ch_\bbK\big)\,:=
\,\Hom_{\Op}^{}\big(\O_{\ovr{\CC}}^{},\Ch_\bbK\big)
\end{flalign}
of $\Ch_\bbK$-valued AQFTs carries
a model structure in which a morphism $\zeta : \AAA\to\BBB$ 
is a weak equivalence (respectively, fibration) if each component 
$\zeta_M^{} : \AAA(M) \to\BBB(M)$ is a quasi-isomorphism (respectively, degree-wise surjective).
We call this the {\em projective model structure} on $\AQFT(\ovr{\CC})$.
\end{theo}
\begin{rem}
Note that every object in $\AQFT(\ovr{\CC})$ is a fibrant object. 
We will frequently use this fact without further emphasis in our arguments below. 
This fact is not essential, but it simplifies some arguments.
\end{rem}

An immediate but fundamental consequence of such model structures 
is given by the following result, which encodes numerous universal 
constructions among AQFTs.
\begin{propo}\label{prop:Quillenprojective}
For each orthogonal functor $F : \ovr{\CC}\to \ovr{\DD}$, the adjunction 
from Proposition \ref{prop:LKE}, i.e.\
$F_! : \AQFT(\ovr{\CC}) \rightleftarrows \AQFT(\ovr{\DD}): F^\ast$,
is a Quillen adjunction for the projective model structures from Theorem \ref{theo:projmodel}.
\end{propo}
\begin{proof}
The right adjoint $F^\ast = (-)\O_{F}^{}$ is a pullback functor, hence it preserves both the fibrations
and the weak equivalences as these are defined component-wise. This in particular implies that 
$F^\ast$ is a right Quillen functor.
\end{proof}

The implementation of the time-slice axiom in this homotopical context 
is more subtle than in the case of an ordinary target category $\TT$ that 
we have discussed in Subsection \ref{subsec:localization}.
Demanding that a multifunctor $\AAA : \O_{\ovr{\CC}}^{}\to \Ch_\bbK$
sends a subset $W\subseteq \Mor\,\CC$ to {\em isomorphisms} is neither
realized in concrete examples (see e.g.\ the linear quantum gauge theories constructed in \cite{LinearYM,BMS})
nor is it compatible with weak equivalences: Indeed, given a weak equivalence
$\zeta :\AAA\to \BBB$ in the model category $\AQFT(\ovr{\CC})$, it is 
in general {\em not} true that $\AAA$ satisfies this property if and only if $\BBB$ does.
The homotopically correct generalization of the time-slice axiom, which is stable under
weak equivalences of AQFTs, is given
by replacing the concept of isomorphisms with that of quasi-isomorphisms.
\begin{defi}\label{def:homotopytimeslice}
Let $\ovr{\CC}$ be an orthogonal category and $W\subseteq \Mor\,\CC$ a subset.
An object $\AAA\in \AQFT(\ovr{\CC})$ is said to satisfy the {\em homotopy time-slice
axiom} for $W$ if, as a multifunctor $\AAA : \O_{\ovr{\CC}}^{}\to \Ch_\bbK$, it sends
$W$ to quasi-isomorphisms. 
\end{defi}

It is important to stress that this definition introduces the
homotopy time-slice axiom only at the level of the objects in the 
model category $\AQFT(\ovr{\CC})$, but it does not define
their morphisms and (higher) homotopies. Since model structures
do not in general restrict in a sensible way to full subcategories, 
more work is required to obtain a model category that presents
the homotopy theory of AQFTs satisfying the homotopy time-slice axiom.
This issue has been addressed and solved in \cite[Section 3.2]{Carmona},
where two different but equivalent approaches have been identified.
The first approach, which we shall sketch only briefly as it will not play an important
role in our paper, uses the concept of {\em homotopical localization of operads},
generalizing Dwyer-Kan localization of simplicial categories \cite{Hinich3}, 
in order to invert (up to homotopy) the $1$-ary 
operations $W$ in the AQFT operad $\O_{\ovr{\CC}}^{}$. Homotopical localizations
can be determined in the model category $\mathbf{sOp}$ of
simplicial (i.e.\ $\sSet$-enriched) colored operads and they are defined by
the following homotopical generalization of the universal property
from Definition \ref{def:operadlocalization} and Remark \ref{rem:Oplocalization}:
A $\mathbf{sOp}$-morphism $L^\infty : \P\to\P[W^{-1}]^{\infty}$ is called
a homotopical localization of $\P\in\mathbf{sOp}$ at a subset $W$ of 
$1$-ary operations in $\P$ if $L^\infty$ sends $W$ to equivalences and, for every
$\Q\in\mathbf{sOp}$, the pullback map
\begin{flalign}\label{eqn:homotopicaluniversal}
(-)L^\infty\,:\,\mathrm{Map}_{\mathbf{sOp}}\big(\P[W^{-1}]^{\infty},\Q\big)~\longrightarrow~\mathrm{Map}_{\mathbf{sOp}}(\P,\Q)^{\mathrm{ho}W}
\end{flalign}
is a weak homotopy equivalence, where $\mathrm{Map}_{\mathbf{sOp}}(\P[W^{-1}]^{\infty},\Q)\in\sSet$
is the mapping space in $\mathbf{sOp}$ and $\mathrm{Map}_{\mathbf{sOp}}(\P,\Q)^{\mathrm{ho}W}\subseteq 
\mathrm{Map}_{\mathbf{sOp}}(\P,\Q)$ denotes the simplicial subset of maps that send $W$ to equivalences.
Homotopical localizations of operads can be computed
by a homotopy pushout that is similar to the one for simplicial categories \cite[Section 3]{Hinich3}
or by a generalization of Hammock localization to trees \cite{LocalizationOperad}.
In our AQFT context, the universal property \eqref{eqn:homotopicaluniversal}
implies that the homotopy time-slice axiom from Definition \ref{def:homotopytimeslice}
can be enforced by considering algebras over the homotopically localized AQFT operad 
$\O_{\ovr{\CC}}^{}[W^{-1}]^{\infty}\in\mathbf{sOp}$. Passing over $\Ch_\bbK$-enriched colored operads
by using the normalized chains functor $N_\ast(-,\bbK) : \sSet\to \Ch_\bbK$, we can define the model
category
\begin{flalign}\label{eqn:enrichedHom}
\AQFT(\ovr{\CC})^{\mathrm{ho}W}\,:=\, \Alg_{\O_{\ovr{\CC}}^{}[W^{-1}]^\infty}^{}\big(\Ch_\bbK\big)
\,:=\,\Hom_{\Op_{\Ch_\bbK}^{}}\big(N_\ast\big(\O_{\ovr{\CC}}^{}[W^{-1}]^\infty,\bbK\big),\Ch_{\bbK}\big)
\end{flalign}
of $\Ch_\bbK$-enriched multifunctors, endowed with the projective model structure from \cite{Hinich1,Hinich2}.
This perspective on the homotopy time-slice axiom is conceptually very clear
thanks to the universal property of homotopical localizations \eqref{eqn:homotopicaluniversal},
but unfortunately it is difficult to use in practice because it is hard to compute
homotopical localizations such as $\O_{\ovr{\CC}}^{}[W^{-1}]^\infty$.
\sk

In order to prove our results in this paper, it will be convenient
to use the equivalent approach from \cite[Section 3.2]{Carmona} that enforces 
the homotopy time-slice axiom via a {\em left Bousfield localization}. 
(We refer the reader to \cite{Barwick} or to Hirschhorn's book \cite{Hirschhorn}
for an introduction to Bousfield localizations of model categories 
and the related concepts of local objects and local equivalences.) 
The basic idea is to work with the same underlying category $\AQFT(\ovr{\CC})$,
but to introduce a new model structure that has more weak equivalences and fewer fibrations
than the projective one in Theorem \ref{theo:projmodel}. When designed correctly,
this will lead to a new model category $\mathcal{L}_{\widehat{W}} \AQFT(\ovr{\CC})$ 
that is Quillen equivalent to \eqref{eqn:enrichedHom} and 
in which the homotopy time-slice axiom becomes a fibrancy condition.
The construction of such a model structure is slightly technical and the relevant details
can be found in \cite{Carmona} and \cite[Section 6]{CFM}. The main result may be 
summarized as follows.
\begin{theo}\label{theo:Bousfield}
Let $\ovr{\CC}$ be an orthogonal category and $W\subseteq \Mor\,\CC$ a subset.
Denote by $\AQFT(\ovr{\CC})$ the projective AQFT model category from Theorem \ref{theo:projmodel}.
\begin{itemize}
\item[a)] There exists a subset $\widehat{W}\subseteq \Mor\,\AQFT(\ovr{\CC})$ 
(see Remark \ref{rem:Whatmaps} below for an explicit description)
such that an object $\AAA\in \AQFT(\ovr{\CC})$
is $\widehat{W}$-local if and only if it satisfies the homotopy time-slice axiom for  $W\subseteq \Mor\,\CC$.

\item[b)] The left Bousfield localization $\mathcal{L}_{\widehat{W}} \AQFT(\ovr{\CC})$ 
of the projective AQFT model category at  $\widehat{W}$ from item a) exists.
The fibrant objects in this localized model structure are precisely the 
AQFTs satisfying the homotopy time-slice axiom for $W$.

\item[c)] The homotopical localization multifunctor $L^\infty : \O_{\ovr{\CC}}^{}\to \O_{\ovr{\CC}}^{}[W^{-1}]^{\infty}$
induces via pullback a Quillen equivalence 
\begin{flalign}
\xymatrix{
L_!^\infty \,:\, \mathcal{L}_{\widehat{W}} \AQFT(\ovr{\CC})\ar@<0.75ex>[r]_-{\sim}~&~\ar@<0.75ex>[l]\AQFT\big(\ovr{\CC})^{\mathrm{ho}W}\,:\,L^{\infty\,\ast}
}
\end{flalign}
between $\mathcal{L}_{\widehat{W}} \AQFT(\ovr{\CC})$ and the projective model category 
$\AQFT\big(\ovr{\CC})^{\mathrm{ho}W}$ in \eqref{eqn:enrichedHom}.
\end{itemize}
\end{theo}
\begin{rem}\label{rem:Whatmaps}
The set $\widehat{W}$ of AQFT morphisms is determined
from the given set $W$ of $\CC$-morphisms
by the following construction from \cite[Definition 6.5 and Proposition 6.4]{CFM}:
Using the covariant Yoneda embedding 
\begin{flalign}
\mathsf{y}(-)\,:\, \CC^{\op} ~\longrightarrow~ \Fun(\CC,\Ch_\bbK)~,~~M~\longmapsto~\mathsf{y}(M)
= \CC(M,-)\otimes \bbK\quad,
\end{flalign}
we can assign to every $\CC$-morphism $f: M\to N$ a 
morphism $\mathsf{y}(f) : \mathsf{y}(N)\to \mathsf{y}(M)$ 
of $\Ch_\bbK$-valued functors on $\CC$. We further can shift
the homological degree by any integer $r\in\mathbb{Z}$ and consider 
$\mathsf{y}(f)[r] : \mathsf{y}(N)[r]\to \mathsf{y}(M)[r]$. Since $i : \CC \hookrightarrow \O_{\ovr{\CC}}^{}$
embeds as the category of $1$-ary operations in the AQFT operad, we obtain in analogy 
to Propositions \ref{prop:LKE} and \ref{prop:Quillenprojective}
a Quillen adjunction $i_! : \Fun(\CC,\Ch_\bbK) \rightleftarrows\AQFT(\ovr{\CC}) : i^\ast$.
The subset $\widehat{W}\subseteq \Mor\,\AQFT(\ovr{\CC})$ then consists of
the morphisms
\begin{flalign}
i_!\big( \mathsf{y}(f)[r]\big)\,:\, i_!\big(\mathsf{y}(N)[r]\big)~\longrightarrow~i_!\big(\mathsf{y}(M)[r]\big)\quad,
\end{flalign}
where $(f: M\to N)\in W$ runs over all morphisms in $W$ and $r\in \bbZ$
runs over all integers.
\end{rem}

\begin{rem}
The reader might wonder about the physical interpretation
of the left Bousfield localized model category
$\mathcal{L}_{\widehat{W}} \AQFT(\ovr{\CC})$ and
its additional weak equivalences. 
In particular, one should {\it not} be surprised that every AQFT $\AAA\in \AQFT(\ovr{\CC})$,
whether or not it satisfies the homotopy time-slice axiom,
is weakly equivalent in $\mathcal{L}_{\widehat{W}} \AQFT(\ovr{\CC})$
to an AQFT that satisfies the homotopy time-slice axiom, e.g.\ by taking a
local fibrant replacement of $\AAA$. We would like to stress
that this behavior is not inconsistent and in fact 
necessary for the model category $\mathcal{L}_{\widehat{W}} \AQFT(\ovr{\CC})$ to be 
Quillen equivalent to the model category $\AQFT(\ovr{\CC})^{\mathrm{ho}W}$ in 
\eqref{eqn:enrichedHom} that describes AQFTs satisfying the homotopy time-slice axiom
through the concept of homotopical localizations, which is precisely the statement of
Theorem \ref{theo:Bousfield} c). In other words,
for the model category $\mathcal{L}_{\widehat{W}} \AQFT(\ovr{\CC})$ 
to describe the homotopy theory of \textit{AQFTs satisfying the homotopy time-slice axiom},
it necessarily must be true that every object in $\mathcal{L}_{\widehat{W}} \AQFT(\ovr{\CC})$ is 
weakly equivalent to an AQFT that satisfies this axiom.
To avoid any misconceptions about the additional weak equivalences in
the model category $\mathcal{L}_{\widehat{W}} \AQFT(\ovr{\CC})$, 
let us recall the following standard result about left Bousfield localizations, 
see e.g.\ \cite[Theorem 3.2.13]{Hirschhorn}: Given any two \textit{fibrant} objects
$\AAA,\BBB\in \mathcal{L}_{\widehat{W}} \AQFT(\ovr{\CC})$ in the localized model
category, a morphism $\zeta :\AAA\to \BBB$
is a local weak equivalence if and only if it is a projective weak equivalence 
in the sense of Theorem \ref{theo:projmodel}. Using Theorem \ref{theo:Bousfield} b) and 
translating this result to AQFT terminology, this means that both types of weak
equivalences (local and projective) coincide whenever $\AAA$ and $\BBB$ satisfy the homotopy time-slice axiom.
In particular, this implies that there are no additional weak equivalences between the objects
that the model category $\mathcal{L}_{\widehat{W}} \AQFT(\ovr{\CC})$ is designed to describe.
\end{rem}

The following result will be crucial for formulating and 
proving our strictification theorems for the homotopy time-slice axiom.
\begin{propo}\label{prop:localizedadjunction}
Suppose that $L : \ovr{\CC}\to \ovr{\CC}[W^{-1}]$ is a localization of orthogonal categories.
The associated adjunction from Proposition \ref{prop:LKE} defines a Quillen adjunction
\begin{flalign}
\xymatrix{
L_! \,:\, \mathcal{L}_{\widehat{W}}\AQFT(\ovr{\CC})\ar@<0.5ex>[r]~&~\ar@<0.5ex>[l]
\AQFT\big(\ovr{\CC}[W^{-1}]\big)\,:\,L^\ast
}
\end{flalign}
between the Bousfield localized model category $\mathcal{L}_{\widehat{W}}\AQFT(\ovr{\CC})$ 
from Theorem \ref{theo:Bousfield} and the projective model category $\AQFT\big(\ovr{\CC}[W^{-1}]\big)$ 
of AQFTs on $\ovr{\CC}[W^{-1}]$ from Theorem \ref{theo:projmodel}.
\end{propo}
\begin{proof}
From the equivalent characterizations of Quillen adjunctions, see e.g.\ \cite[Proposition 8.5.3]{Hirschhorn},
we find it most convenient to prove that $L_!$ preserves cofibrations and that $L^\ast$ preserves fibrations.
Recalling that the cofibrations in the left Bousfield localization  
$\mathcal{L}_{\widehat{W}}\AQFT(\ovr{\CC})$ agree by \cite[Definition 3.3.1]{Hirschhorn}
with the cofibrations in the projective model category $\AQFT(\ovr{\CC})$, the first claim follows directly 
from Proposition \ref{prop:Quillenprojective}. 
\sk

To prove that $L^\ast$ preserves fibrations, let us first observe that, for every 
$\BBB\in \AQFT\big(\ovr{\CC}[W^{-1}]\big)$, the AQFT $L^\ast (\BBB)\in \mathcal{L}_{\widehat{W}}\AQFT(\ovr{\CC})$ 
satisfies the strict time-slice axiom and hence also the homotopy time-slice axiom. 
Recalling Theorem \ref{theo:Bousfield} b), this means that $L^\ast$ maps to the
fibrant objects in $\mathcal{L}_{\widehat{W}}\AQFT(\ovr{\CC})$.
Combining this with Proposition \ref{prop:Quillenprojective}, we obtain that
$L^\ast$ maps fibrations to projective fibrations between fibrant 
objects in $\mathcal{L}_{\widehat{W}}\AQFT(\ovr{\CC})$, which due to
\cite[Proposition 3.3.16]{Hirschhorn} are also fibrations in the localized model structure.
\end{proof}
\begin{rem}
Observe that the Quillen adjunction from Proposition \ref{prop:localizedadjunction} allows us to compare between 
the homotopy time-slice axiom from Definition \ref{def:homotopytimeslice} 
and the usual {\em strict} time-slice axiom that is formulated in terms of isomorphisms rather than quasi-isomorphisms.
Indeed, by Theorem \ref{theo:Bousfield} c), 
the model category $\mathcal{L}_{\widehat{W}}\AQFT(\ovr{\CC})$ describes
AQFTs that satisfy the homotopy time-slice axiom and, by Corollary \ref{cor:TimeSlice via LocalizationOfOrtCat},
the model category $\AQFT\big(\ovr{\CC}[W^{-1}]\big)$ describes AQFTs that satisfy the strict time-slice axiom.
\end{rem}


\section{\label{sec:reflectivestrictification}Strictification theorem for reflective localizations}
In this section we consider a special class of orthogonal localizations
that are reflective in the sense of Definition \ref{def:OCatreflocalization}.
In this context we shall prove that the homotopy time-slice axiom 
is equivalent (in a suitable sense) to the strict time-slice axiom. 
The precise statement is our Theorem \ref{theo:strictificationreflective} below.
\sk

The following definition is an adaption of the concept of 
reflective localizations from ordinary category theory.
\begin{defi}\label{def:OCatreflocalization}
A localization of orthogonal categories $L : \ovr{\CC}\to \ovr{\CC}[W^{-1}]$
is called {\em reflective} if the orthogonal functor $L$ admits a right adjoint 
$\iota : \ovr{\CC}[W^{-1}] \to \ovr{\CC}$ in the $2$-category $\OCat$
whose underlying functor $\iota : \CC[W^{-1}]\to \CC$ is fully faithful.
\end{defi}
\begin{rem}
From this definition, it follows that
the orthogonality relation $\perp_{\CC[W^{-1}]}^{}$ agrees with the pullback 
along $\iota$ of  $\perp_{\CC}^{}$, 
i.e.\ $\perp_{\CC[W^{-1}]}^{}\,=\,\iota^\ast(\perp_{\CC}^{})$.
The inclusion $\perp_{\CC[W^{-1}]}^{}\,\subseteq \iota^{\ast}(\perp_{\CC}^{})$
is a consequence of $\iota$ being an orthogonal functor
and the inclusion $\perp_{\CC[W^{-1}]}^{}\,\supseteq \iota^{\ast}(\perp_{\CC}^{})$
is proven by using that $\iota$ is fully faithful. 
Indeed, given any $(g_1,g_2)\in \iota^{\ast}(\perp_{\CC}^{})$, 
we have by definition of the pullback orthogonality relation that 
$(\iota(g_1),\iota(g_2))\in \,\perp_{\CC}^{}$, 
hence $(L\iota(g_1),L\iota(g_2))\in\, \perp_{\CC[W^{-1}]}^{}$
because $L$ is an orthogonal functor. Using that $\perp_{\CC[W^{-1}]}^{}$
is composition stable and that the adjunction counit $\epsilon: L \iota\to \id_{\CC[W^{-1}]}^{}$ 
is a natural isomorphism,
it follows that $(g_1,g_2)\in\,\perp_{\CC[W^{-1}]}^{}$.
\end{rem}

The following result is useful to detect reflective localizations of orthogonal categories.
\begin{propo}\label{propo:OCatdetectlocalization}
Suppose that $L : \ovr{\CC}\to\ovr{\EE}$ is an orthogonal functor
that admits a fully faithful right adjoint $j : \ovr{\EE}\to \ovr{\CC}$ 
in the $2$-category $\OCat$. Then $L : \ovr{\CC}\to\ovr{\EE}$ is a reflective
localization of the orthogonal category $\ovr{\CC}$ at the set of morphisms
$W:= L^{-1}(\mathrm{Iso}\,\EE)\subseteq \Mor\,\CC$ that are sent via $L$ to isomorphisms in $\ovr{\EE}$.
\end{propo}
\begin{proof}
It is a basic exercise in category theory to prove that the underlying functor
$L: \CC\to \EE$ is a localization of the category $\CC$ at $W= L^{-1}(\mathrm{Iso}\,\EE)\subseteq \Mor\,\CC$.
The hypothesis that $L : \ovr{\CC}\to\ovr{\EE}$ is an orthogonal functor
implies that $L_\ast(\perp_{\CC}^{})\subseteq \,\perp_{\EE}^{}$, so it remains
to prove the inclusion $\perp_{\EE}^{}\, \subseteq L_\ast(\perp_{\CC}^{})$.
Given any $g_1\perp_{\EE}^{}g_2$, it follows that $j(g_1)\perp_{\CC}^{}j(g_2)$
and also $Lj(g_1)\perp_{\EE}^{}Lj(g_2)$ because $j$ and $L$ are orthogonal functors.
Note that $\big(Lj(g_1),Lj(g_2)\big)\in L_\ast(\perp_{\CC}^{})\subseteq \,\perp_{\EE}^{}$
defines an element of the pushforward orthogonality relation. Using that $L_\ast(\perp_{\CC}^{})$
is composition stable and that the adjunction counit $\epsilon: L j\to \id_{\EE}^{}$ is a natural isomorphism,
it follows that $(g_1,g_2)\in L_\ast(\perp_{\CC}^{})$.
\end{proof}

\begin{ex}\label{ex:OCatreflocalization}
Of particular interest for AQFT is the case where $\ovr{\CC}$ is some 
category of globally hyperbolic Lorentzian manifolds, with orthogonality relation determined by
causal disjointness, and $W\subseteq \Mor\,\CC$ is the subset of all Cauchy morphisms. See e.g.\ 
\cite{BFV,FewsterVerch} and also \cite[Section 3.5]{BSWoperad} for the relevant context and terminology.
The following examples of reflective localizations have been established in the literature:
\begin{itemize}
\item In \cite[Proposition 3.3]{BDSboundary}, it is shown that the localization at all Cauchy morphisms
of the orthogonal category $\ovr{\mathcal{R}_M}$ of causally convex open subsets of a fixed globally 
hyperbolic Lorentzian manifold $M$ (with or without time-like boundary) is reflective.
This family of examples covers Haag-Kastler-type AQFTs.

\item In \cite[Section 3]{Skeletal}, it is shown that the localization at all Cauchy morphisms
of the orthogonal category $\ovr{\mathbf{CLoc}_2}$ of oriented and time-oriented connected globally hyperbolic {\em conformal} 
Lorentzian $2$-manifolds is reflective.  This example covers locally covariant {\em conformal} AQFTs, in the sense of
\cite{Pinamonti}, in two spacetime dimensions.
\end{itemize}

Another example is the localization of the 
orthogonal category $\ovr{\Loc_1}$ of oriented and time-oriented
connected globally hyperbolic Lorentzian $1$-manifolds at all Cauchy morphisms. 
As this example has not been recorded in the literature yet, we include a proof in Appendix \ref{app:Loc1}.
On the other hand, it is presently unclear to us whether or not 
the localization at all Cauchy morphisms of the orthogonal category $\ovr{\Loc_m}$ of oriented and time-oriented
connected globally hyperbolic Lorentzian $m$-manifolds is reflective for dimension $m\geq 2$.
\end{ex}

\begin{cex}\label{cex:disksreflective}
This counterexample is inspired by topological field theories.
Denote by $\Disk(\bbR^m)$ the category whose 
objects are open $m$-disks $U\subseteq \bbR^m$ and morphisms are subset inclusions.
We endow this category with the disjointness orthogonality relation, i.e.\ 
$(U_1\subseteq V)\perp (U_2\subseteq V)$ if and only if $U_1\cap U_2 = \emptyset$.
Let us take $W$ to be the isotopy equivalences, which are
all morphisms in $\Disk(\bbR^m)$. The localization $L : \Disk(\bbR^m) \to \Disk(\bbR^m)[W^{-1}]=\{\ast\}$ 
of the underlying category is a singleton, i.e.\ a category with only one object 
$\ast$ and its identity morphism $\id_\ast$.  The induced orthogonality relation
is given by $L_\ast(\perp) = \{(\id_\ast,\id_\ast)\}$, i.e.\ the identity $\id_\ast$ is orthogonal to itself.
At the level of the underlying categories, this localization is reflective with right adjoint
$\iota : \{\ast\} \to \Disk(\bbR^m)$ defined by $\iota(\ast) = \bbR^m$. However, the functor $\iota$ 
is {\em not} orthogonal because $\bbR^m\cap \bbR^m\neq \emptyset$ are not disjoint. 
This provides a simple example of an orthogonal localization 
$L : \ovr{\Disk(\bbR^m)} \to \ovr{\{\ast\}}$ that is not reflective
in the sense of Definition \ref{def:OCatreflocalization}, even though the underlying 
localization of categories is reflective. 
\end{cex}

Given any reflective localization of orthogonal categories
$L : \ovr{\CC}\rightleftarrows \ovr{\CC}[W^{-1}] : \iota$,
we obtain from the $2$-functor \eqref{eqn:AQFT2functor} an adjunction
\begin{flalign}\label{eqn:reflectiveleftadjoint}
\xymatrix{
\iota^\ast \,:\, \AQFT(\ovr{\CC})\ar@<0.5ex>[r]~&~\ar@<0.5ex>[l]
\AQFT\big(\ovr{\CC}[W^{-1}]\big)\,:\,L^\ast
}
\end{flalign}
between the associated AQFT categories. From the uniqueness
(up to natural isomorphism) of adjoint functors, it then follows that
$\iota^\ast $ is a model for the left Quillen functor $L_!$ in 
Proposition \ref{prop:localizedadjunction}. 
This is the key observation to prove the following strictification theorem.
\begin{theo}\label{theo:strictificationreflective}
Let $L : \ovr{\CC}\rightleftarrows \ovr{\CC}[W^{-1}] : \iota$ be a 
reflective localization of orthogonal categories.
Then the Quillen adjunction
\begin{flalign}
\xymatrix{
L_!\,=\,\iota^\ast \,:\, \mathcal{L}_{\widehat{W}}\AQFT(\ovr{\CC})\ar@<0.5ex>[r]~&~\ar@<0.5ex>[l]
\AQFT\big(\ovr{\CC}[W^{-1}]\big)\,:\,L^\ast
}
\end{flalign}
from Proposition \ref{prop:localizedadjunction} is a Quillen equivalence.
\end{theo}
\begin{proof}
Let us fix any cofibrant replacement $(Q,q)$ for the {\em projective} model structure
on $\AQFT(\ovr{\CC})$, which by \cite[Definition 3.3.1]{Hirschhorn} determines 
a cofibrant replacement for the left Bousfield localized model category
$\mathcal{L}_{\widehat{W}}\AQFT(\ovr{\CC})$. Let us also 
recall that every object in the projective model category $\AQFT\big(\ovr{\CC}[W^{-1}]\big)$ is fibrant.
We then may take $\bbR L^\ast := L^\ast$ and $\bbL \iota^\ast := \iota^\ast\,Q$
as models for the derived functors.
\sk

To show that the derived unit is a natural weak equivalence, it suffices to consider its components
on those objects $\AAA\in\mathcal{L}_{\widehat{W}}\AQFT(\ovr{\CC})$ that are both fibrant and cofibrant.
The question then reduces to proving that the associated components
$\eta_{\AAA}^{} : \AAA\to L^\ast \iota^\ast(\AAA)$ of the ordinary unit
are weak equivalences in $\mathcal{L}_{\widehat{W}}\AQFT(\ovr{\CC})$, or equivalently
projective weak equivalences because both $\AAA$ and $L^\ast \iota^\ast(\AAA)$ are fibrant objects, 
see e.g.\ \cite[Theorem 3.2.13]{Hirschhorn}.
Recalling that the adjunction \eqref{eqn:reflectiveleftadjoint} is obtained
through the AQFT $2$-functor \eqref{eqn:AQFT2functor}, we have that $\eta_{\AAA}^{} =\AAA\,\O_{\eta^{\perp}}^{}  
: \AAA\to L^\ast \iota^\ast(\AAA)= \AAA\,\O_{\iota L}$
where $\eta^\perp$ is the unit for the adjunction $L : \ovr{\CC}\rightleftarrows \ovr{\CC}[W^{-1}] : \iota$ in $\OCat$.
By Proposition \ref{propo:OCatdetectlocalization}, we may assume without loss of 
generality that $W = L^{-1}(\mathrm{Iso}\,\CC[W^{-1}])$, which combined with 
the triangle identities for the adjunction $L\dashv \iota$
implies that all components of $\eta^\perp$ are morphisms in $W$. It then follows that
$\eta_{\AAA}^{} $ is a projective weak
equivalence because $\AAA$ satisfies the homotopy time-slice axiom as it is by hypothesis a fibrant object, 
see Theorem \ref{theo:Bousfield} b). 
\sk

Concerning the derived counit, let us consider its component
\begin{flalign}
\xymatrix{
\iota^\ast\, Q\,L^\ast(\BBB) \ar[rr]^-{\iota^\ast q_{L^\ast(\BBB)}^{}} ~&&~ \iota^\ast\,L^\ast(\BBB)  \ar[r]^-{\epsilon_{\BBB}^{}} ~&~ \BBB
}
\end{flalign}
on an arbitrary object $\BBB\in \AQFT\big(\ovr{\CC}[W^{-1}]\big)$. By our choice
of cofibrant replacement, $q_{L^\ast(\BBB)}^{}$ is a projective weak equivalence 
and hence the first arrow is a weak equivalence since $\iota^\ast$ preserves projective weak equivalences.
It thus remains to show that the second arrow is a weak equivalence too.
For this we recall once more that the adjunction \eqref{eqn:reflectiveleftadjoint} is 
obtained through the AQFT $2$-functor \eqref{eqn:AQFT2functor}, hence $\epsilon_{\BBB}^{} =\BBB\,\O_{\epsilon^\perp}^{} 
:\BBB\, \O_{L\iota}^{}=\iota^\ast\,L^\ast(\BBB)\to \BBB$. The claim then follows
from the fact that the counit $\epsilon^\perp$ 
for the adjunction $L : \ovr{\CC}\rightleftarrows \ovr{\CC}[W^{-1}] : \iota$ in $\OCat$
is a natural isomorphism since the right adjoint $\iota$ is fully faithful.
\end{proof}

A direct consequence of Theorem \ref{theo:strictificationreflective}
is that each $\AAA\in \AQFT(\ovr{\CC})$ that satisfies the homotopy 
time-slice axiom for $W$ is \textit{projectively} weakly equivalent to some 
$\AAA_{\mathrm{st}}\in \AQFT(\ovr{\CC})$ that satisfies the 
strict time-slice axiom for $W$. Such an object can be constructed 
very explicitly in the present case using only {\em underived} functors, 
i.e.\ there is no need to develop models for derived functors.
This important and useful fact is summarized in the following
\begin{cor}\label{cor:reflectiveAQFTwise}
Let $\AAA\in \AQFT(\ovr{\CC})$ be any AQFT that satisfies the homotopy 
time-slice axiom for $W$. Then the corresponding
component $\eta_{\AAA}^{} : \AAA \to L^\ast \iota^\ast(\AAA)$
of the underived unit is a projective weak equivalence,
i.e.\ a weak equivalence in the usual projective model structure 
from Theorem \ref{theo:projmodel}. In particular, 
$\AAA_{\mathrm{st}} :=  L^\ast\iota^\ast(\AAA)$
defines an AQFT that satisfies the strict time-slice axiom 
and is projectively weakly equivalent to $\AAA$.
\end{cor}
\begin{proof}
It suffices to repeat the second paragraph of the proof of Theorem \ref{theo:strictificationreflective} 
(note that the cofibrancy assumption plays no role there). 
\end{proof}

\begin{ex}\label{ex:explicitlistofexamples}
Recalling our list of reflective orthogonal localizations from Example \ref{ex:OCatreflocalization}
and Appendix \ref{app:Loc1}, we would like to emphasize that our 
Theorem \ref{theo:strictificationreflective}, which strictifies the homotopy time-slice axiom, covers 
the following important cases: (i)~Haag-Kastler-type AQFTs on a fixed globally 
hyperbolic Lorentzian manifold (with or without time-like boundary), (ii)~locally covariant
conformal AQFTs in two spacetime dimensions, and (iii)~locally covariant AQFTs in one spacetime dimension.
\end{ex}


\section{\label{sec:emptystrictification}Strictification criteria for empty orthogonality relations}
In this section we consider the case of an orthogonal category 
$\ovr{\CC} = (\CC,\emptyset)$ that carries an empty orthogonality relation
and establish criteria under which there exists a 
strictification theorem for the associated homotopy time-slice axiom.
This case is motivated by its relevance for describing the
{\em relative Cauchy evolution} (RCE) for $\Ch_\bbK$-valued AQFTs, 
see \cite{BFSRCE} for earlier results that we will generalize 
in Example \ref{ex:RCE} below.
\sk

The main simplification that arises in the case of an empty orthogonality relation
is that the category of $\Ch_\bbK$-valued AQFTs from Theorem \ref{theo:projmodel}
can be presented as the category
\begin{flalign}\label{eqn:AQFTFunCat}
\AQFT(\CC,\emptyset) \,\cong\, \Fun\big(\CC,\Alg_{\mathsf{As}}^{}(\Ch_\bbK)\big)
\end{flalign}
of functors from $\CC$ to associative and unital differential graded algebras.
The projective model structure  on $\AQFT(\CC,\emptyset)$ from Theorem \ref{theo:projmodel} then
gets identified with the projective model structure on $\Fun\big(\CC,\Alg_{\mathsf{As}}^{}(\Ch_\bbK)\big)$.
\sk

Using Proposition \ref{prop:Olocalizationmodel}, we obtain that the orthogonal localization
$L : (\CC,\emptyset)\to (\CC[W^{-1}],\emptyset)$ of $\ovr{\CC}=(\CC,\emptyset)$ 
at any set of morphisms $W\subseteq \Mor\,\CC$ is given by the usual categorical localization $L$
and that the target category carries an empty orthogonality relation $L_\ast(\emptyset) = \emptyset$ too.
The key tool that we will use in order to simplify the strictification problem of AQFTs
is the following square of Quillen adjunctions
\begin{flalign}\label{eqn:squareadjunction}
\xymatrix@R=4em@C=5em{
\ar@<0.75ex>[d]^-{i^\ast} \mathcal{L}_{\widehat{W}}\AQFT(\CC,\emptyset) \ar@<0.75ex>[r]^-{L_!}~&~ \ar@<0.75ex>[l]^-{L^\ast} \AQFT\big(\CC[W^{-1}],\emptyset\big)\ar@<0.75ex>[d]^-{j^\ast}\\
\ar@<0.75ex>[u]^-{i_!} \mathcal{L}_{\widetilde{W}}\Fun(\CC,\Ch_\bbK) \ar@<0.75ex>[r]^-{\Lan_L^{}}~&~\ar@<0.75ex>[l]^-{L^\ast}\Fun\big(\CC[W^{-1}],\Ch_\bbK\big) \ar@<0.75ex>[u]^-{j_!}
}
\end{flalign}
The top horizontal adjunction is the one controlling the strictification problem for AQFTs,
see Proposition \ref{prop:localizedadjunction}. The bottom horizontal adjunction is its analogue for
$\Ch_\bbK$-valued functors, where $\Lan_L^{}$ denotes the left Kan extension along the underlying 
functor $L:\CC\to\CC[W^{-1}]$. The left Bousfield localization 
$\mathcal{L}_{\widetilde{W}}\Fun(\CC,\Ch_\bbK)$ is defined in analogy to Remark \ref{rem:Whatmaps} 
with the set of maps  $\widetilde{W}$ given by $\mathsf{y}(f)[r]$, for all $(f:M\to N)\in W$ and all integers $r\in\bbZ$,
i.e.\ without applying the functor $i_!$. 
The right adjoints $i^\ast$ and $j^\ast$ 
of the vertical adjunctions are the functors that forget the multiplications and units of an AQFT.
Their left adjoints $i_!$ and $j_!$ are, in the present case of empty orthogonality 
relations (see \eqref{eqn:AQFTFunCat}), simply the functors that take object-wise free algebras.
For later use, we would like to record the following commutativity properties
\begin{flalign}\label{eqn:squarecomrels}
i^\ast\,L^\ast\,=\,L^\ast\,j^\ast\quad,\qquad
L_!\,i_!\,\cong\, j_!\,\Lan_L^{}\quad,\qquad
i_!\,L^\ast \,=\, L^\ast\,j_!\quad
\end{flalign}
of the functors in the square of adjunctions \eqref{eqn:squareadjunction}.
Note that the last property follows from the fact that $i_!$ and $j_!$ 
form object-wise free algebras, which commutes with pullback functors, 
hence it is special to the case of empty orthogonality relations.
\sk

The following result reduces the strictification problem for the homotopy time-slice
axiom of AQFTs with empty orthogonality relation to a simpler, but still non-trivial, 
strictification problem for $\Ch_\bbK$-valued functors.
\begin{theo}\label{theo:emptystrictification}
Let $\ovr{\CC} = (\CC,\emptyset)$ be an orthogonal category with empty 
orthogonality relation and $W\subseteq \Mor\,\CC$ a subset.
Suppose that the Quillen adjunction 
\begin{flalign}\label{eqn:Lanadjunction}
\xymatrix{
\Lan_L^{} \,:\, \mathcal{L}_{\widetilde{W}}\Fun(\CC,\Ch_\bbK)\ar@<0.5ex>[r]~&~\ar@<0.5ex>[l]
\Fun\big(\CC[W^{-1}],\Ch_\bbK\big)\,:\,L^\ast
}
\end{flalign}
for $\Ch_\bbK$-valued functors
is a Quillen equivalence. Then the Quillen adjunction
\begin{flalign}
\xymatrix{
L_!\,:\, \mathcal{L}_{\widehat{W}}\AQFT(\CC,\emptyset)\ar@<0.5ex>[r]~&~\ar@<0.5ex>[l]
\AQFT\big(\CC[W^{-1}],\emptyset \big)\,:\,L^\ast
}
\end{flalign}
from Proposition \ref{prop:localizedadjunction} is a Quillen equivalence too.
\end{theo}
\begin{proof}
To describe the derived functors, we pick 
a cofibrant replacement $(Q,q)$ for $\mathcal{L}_{\widehat{W}}\AQFT(\CC,\emptyset)$
and a cofibrant replacement $(Z,z)$ for $\mathcal{L}_{\widetilde{W}}\Fun(\CC,\Ch_\bbK)$,
such that $q$ and $z$ are natural projective weak equivalences. We have
to prove the following two statements: 
\begin{itemize}
\item[1.)] For all locally fibrant objects 
$\AAA\in \mathcal{L}_{\widehat{W}}\AQFT(\CC,\emptyset)$, i.e.\ $\AAA$ satisfies the homotopy time-slice axiom, 
the derived unit $\eta_{Q(\AAA)}^{} : Q(\AAA)\to L^\ast\,L_!\,Q(\AAA)$ is a local weak equivalence, or equivalently
a projective weak equivalence by \cite[Theorem 3.2.13]{Hirschhorn} because 
$Q(\AAA)$ and $L^\ast\,L_!\,Q(\AAA)$ are locally fibrant objects.

\item[2.)] For all objects $\BBB\in \AQFT\big(\CC[W^{-1}],\emptyset \big)$, the derived counit
\begin{flalign}
\xymatrix@C=1.5em{
L_!\,Q\,L^\ast(\BBB)  \ar[rr]^-{L_!q_{L^\ast(\BBB)}^{}} && L_!\,L^\ast(\BBB)\ar[r] \ar[r]^-{\epsilon_{\BBB}^{}}& \BBB
}
\end{flalign}
is a projective weak equivalence. 
\end{itemize}
Our proof strategy is to use
suitable cotriple resolutions \cite[Chapter 13.3]{Fresse} to reduce
this problem to objects of the form $\AAA=i_!(\FFF)$ and $\BBB = j_!(\GGG)$,
where $\FFF\in  \mathcal{L}_{\widetilde{W}}\Fun(\CC,\Ch_\bbK)$
is any locally fibrant object and $\GGG\in \Fun\big(\CC[W^{-1}],\Ch_\bbK\big)$ is any object.
Leveraging the square of adjunctions \eqref{eqn:squareadjunction}, and in particular
the commutativity properties \eqref{eqn:squarecomrels}, then allows us to deduce,
from the hypothesis that \eqref{eqn:Lanadjunction} is a Quillen equivalence,
that these components of the derived unit/counit are projective weak equivalences.
\sk

Our cotriple resolutions are determined by the vertical adjunctions in \eqref{eqn:squareadjunction}.
Let us consider first the left vertical adjunction.
The endofunctor $T:= i_!\,i^\ast$ on $\mathcal{L}_{\widehat{W}}\AQFT(\CC,\emptyset)$ defines
a comonad with coproduct $i_!\eta^i i^\ast : T = i_!\,i^\ast \to i_!\,i^\ast\,i_!\,i^\ast=T^2$
given by the adjunction unit and counit $\epsilon^i : T = i_!\,i^\ast\to \id$ given 
by the adjunction counit. This allows us to define a simplicial resolution
\begin{flalign}\label{eqn:Res}
\mathrm{Res}(\AAA)\,:=\,
\bigg( \xymatrix@C=1.5em{
T(\AAA) \ar@<0pt>[r] ~&~ T^2(\AAA) \ar@<4pt>[l] \ar@<-4pt>[l] \ar@<4pt>[r] \ar@<-4pt>[r] ~&~ \cdots \ar@<8pt>[l] \ar@<0pt>[l] \ar@<-8pt>[l]
} \bigg)\,\in \,\mathcal{L}_{\widehat{W}}\AQFT(\CC,\emptyset)^{\Delta^\op}
\end{flalign}
and an augmentation map $\mathrm{Res}(\AAA)\to\AAA$. Since both $i_!$ and $i^\ast$ act
object-wise on the underlying category $\CC$, so does $\mathrm{Res}$. In fact,
recalling \eqref{eqn:AQFTFunCat}, the resolution $\mathrm{Res}$ is simply 
an object-wise free resolution of dg-algebra valued functors. 
Using the normalized totalization functor
$\Tot^\oplus : \Alg_{\mathsf{As}}^{}(\Ch_\bbK)^{\Delta^\op} \to \Alg_{\mathsf{As}}^{}(\Ch_\bbK)$,
which is a homotopy colimit functor for simplicial diagrams of dg-algebras \cite{Harper},
the result in \cite[Lemma 13.3.3]{Fresse} implies that
\begin{flalign}
\xymatrix@C=2em{
\Tot^{\oplus}\mathrm{Res}(\AAA)\ar[r]^-{\sim}~&~\AAA
}
\end{flalign}
is a projective weak equivalence. The same holds true 
for the right vertical adjunction in \eqref{eqn:squareadjunction}
by using instead of $T$ the comonad $T^\prime := j_!\,j^\ast$ on $\AQFT\big(\CC[W^{-1}],\emptyset\big)$.
\sk

An important feature of these resolutions is that $\Tot^\oplus$ commutes (up to weak equivalence)
with both the derived right adjoint $L^\ast$ and
the derived left adjoint $L_!\,Q$ of the top horizontal adjunction in \eqref{eqn:squareadjunction}.
The former is a consequence of the fact that $L^\ast$ is a pullback functor
and $\Tot^\oplus$ acts by post-composition on functors, while the latter follows from
the fact that derived left adjoints $L_!\,Q$ commute with homotopy colimits.
Using also the explicit form of the simplicial resolution \eqref{eqn:Res}, this implies that
the question of whether the derived unit $\eta_{Q(\AAA)}^{} : Q(\AAA)\to L^\ast\,L_!\,Q(\AAA)$
is a projective weak equivalence for all locally fibrant $\AAA\in  \mathcal{L}_{\widehat{W}}\AQFT(\CC,\emptyset)$ 
can be reduced to objects of the form $\AAA = T(\AAA^\prime) = i_!\,i^\ast(\AAA^\prime)$,
for all locally fibrant $\AAA^\prime\in \mathcal{L}_{\widehat{W}}\AQFT(\CC,\emptyset)$.
Using that $i_!$ and $i^\ast$ both preserve and detect 
local fibrancy (i.e.\ the homotopy time-slice axiom), 
one can equivalently consider objects of the form $\AAA = i_!(\FFF)$, 
for all locally fibrant $\FFF\in \mathcal{L}_{\widetilde{W}}\Fun(\CC,\Ch_\bbK)$.
Applying the same arguments to the derived counit, one obtains that it is sufficient
to consider its components on objects of the form $\BBB = j_!(\GGG)$, for all 
$\GGG\in \Fun\big(\CC[W^{-1}],\Ch_\bbK\big)$.
\sk

 Using now the commutativity properties \eqref{eqn:squarecomrels}, one obtains the
commutative diagram
\begin{flalign}
\xymatrix@C=3.6em{
Q\, i_!(\FFF) \ar[d]_-{\eta_{Qi_!(\FFF)}^{}} ~&~
Q\, i_!\, Z(\FFF)  \ar[l]^-{\sim}_-{Qi_!z_{\FFF}^{}} \ar[d]_-{\eta_{Qi_!Z(\FFF)}^{}}  \ar[r]_-{\sim}^-{q_{i_!Z(\FFF)}^{}}  ~&~ 
i_! \, Z(\FFF) \ar[d]_-{\eta_{i_!Z(\FFF)}^{}}  \ar[dr]^-{i_! \eta^{\mathrm{Fun}}_{Z(\FFF)}}~&~ 
\\
L^\ast \,L_!\, Q \,i_!(\FFF) ~&~
L^\ast\,  L_!\, Q\, i_!\, Z(\FFF)   \ar[l]_-{\sim}^-{L^\ast L_! Qi_!z_{\FFF}^{}} \ar[r]^-{\sim}_-{L^\ast L_! q_{i_!Z(\FFF)}^{}} ~&~ 
L^\ast \, L_! \, i_! \, Z(\FFF) \ar[r]_-{\cong}~&~ 
i_!\, L^\ast\, \Lan_L^{}\, Z(\FFF)
}
\end{flalign}
that relates the derived units of the top and the bottom horizontal adjunction 
in \eqref{eqn:squareadjunction}. Recall that $(Z,z)$ denotes a projective cofibrant replacement 
for $\mathcal{L}_{\widetilde{W}}\Fun(\CC,\Ch_\bbK)$ and that projective weak equivalences are indicated by $\sim$.
Since by hypothesis the derived unit $\eta^{\mathrm{Fun}}_{Z(\FFF)}$ associated to 
$\Lan_L\dashv L^\ast$ is a projective weak equivalence
and since $i_!$ preserves projective weak equivalences, it follows that $\eta_{Qi_!(\FFF)}^{}$
is a projective weak equivalence too.
Using again \eqref{eqn:squarecomrels}, one obtains the
commutative diagram
\begin{flalign}
\xymatrix@C=4em{
L_!\,Q\,L^\ast\,j_!(\GGG) \ar[d]_-{L_! q_{L^\ast j_!(\GGG)}^{}} ~&~
L_!\,Q\,i_!\,L^\ast(\GGG) \ar[d]_-{L_! q_{i_!L^\ast(\GGG)}^{}} \ar[l]_-{\cong} ~&~
L_!\,Q\,i_!\,Z\,L^\ast(\GGG)  \ar[d]^-{L_! q_{i_!ZL^\ast(\GGG)}^{}}_-{\sim} \ar[l]_-{L_!Qi_!z_{L^\ast(\GGG)}^{}}^-{\sim} \\
L_!\,L^\ast\,j_!(\GGG) \ar[d]_-{\epsilon_{j_!(\GGG)}^{}}~&~
L_!\,i_!\,L^\ast(\GGG) \ar[l]_-{\cong} \ar[d]_-{\cong}~&~
L_!\,i_!\,Z\,L^\ast(\GGG) \ar[l]_-{L_! i_!z_{L^\ast(\GGG)}^{}} \ar[d]^-{\cong} \\
j_!(\GGG) ~&~
j_!\,\Lan_L^{}\,L^\ast(\GGG) \ar[l]^-{j_!\epsilon^{\mathrm{Fun}}_{\GGG}}~&~
j_!\,\Lan_L^{}\,Z\,L^\ast(\GGG)  \ar[l]^-{j_!\Lan_L z_{L^\ast(\GGG)}^{}}
}
\end{flalign}
that relates the derived counits of the top and the bottom horizontal adjunctions in \eqref{eqn:squareadjunction}.
The bottom horizontal composition is the application of the functor $j_!$ 
to the derived counit associated to $\Lan_L\dashv L^\ast$.
Since by hypothesis the latter is a projective weak equivalence and since $j_!$ preserves projective weak equivalences,
it follows that the left vertical composition is a projective weak equivalence too. This completes the proof.
\end{proof}

\begin{rem}\label{rem:emptyAQFTwise}
In analogy to Corollary \ref{cor:reflectiveAQFTwise}, 
a direct consequence of Theorem \ref{theo:emptystrictification}
is that each $\AAA\in \AQFT(\CC,\emptyset)$ that satisfies the homotopy 
time-slice axiom for $W$ is \textit{projectively} weakly equivalent to some
$\AAA_{\mathrm{st}}\in \AQFT(\CC,\emptyset)$ that satisfies the 
strict time-slice axiom for $W$. However, the construction
of such an object is more complicated in the present case
because it requires derived functors: Fixing as in
the proof of Theorem \ref{theo:emptystrictification} any cofibrant
replacement $(Q,q)$ for the projective model structure on 
$\AQFT(\CC,\emptyset)$ and recalling that every object in 
$\AQFT\big(\CC[W^{-1}],\emptyset \big)$ is fibrant, 
we may take $\bbR L^\ast := L^\ast$ and $\bbL L_! := L_!\,Q$ as models
for the derived functors. The component at $\AAA\in \AQFT(\CC,\emptyset)$
of the derived unit of the Quillen adjunction $L_!\dashv L^\ast$
is then given by the zig-zag
\begin{flalign}\label{eqn:zigzag}
\xymatrix{
\AAA ~&~ \ar[l]_-{q_{\AAA}^{}} Q(\AAA) \ar[r]^-{\eta_{Q(\AAA)}^{}}~&~ L^\ast\,L_!\,Q(\AAA)
}\quad.
\end{flalign}
The map $q_{\AAA}^{}$ is by construction a projective weak equivalence
and, as explained in the proof of Theorem \ref{theo:emptystrictification}, 
so is the map $\eta_{Q(\AAA)}$.
This implies that setting $\AAA_{\mathrm{st}}:= L^\ast\,\bbL L_!(\AAA) = L^\ast\,L_!\,Q(\AAA)$ 
defines an AQFT that satisfies the strict time-slice axiom 
and is equivalent to $\AAA$ via the zig-zag of \textit{projective} weak equivalences \eqref{eqn:zigzag}.
To obtain a computable model for $\AAA_{\mathrm{st}}$, one can use \cite[Theorem 17.2.7]{Fresse} 
to describe, up to further projective weak equivalences,
\begin{flalign}
\AAA_{\mathrm{st}} = L^\ast\,\bbL L_!(\AAA)\simeq L^\ast\,\Tot^{\oplus}B_{\Delta}\big(\O_{(\CC[W^{-1}],\emptyset)},\O_{(\CC,\emptyset)}^{},\AAA\big)
\end{flalign}
in terms of an operadic cotriple resolution, which can be worked out fairly explicitly
using the definitions in \cite[Chapter 13.3]{Fresse}. (Note that the
normalized totalization functor $\Tot^\oplus$ is denoted by $N_\ast$ in this book.)
\end{rem}

\begin{rem}\label{rem:inftylocalization}
A sufficient condition for the hypotheses of 
Theorem \ref{theo:emptystrictification} to hold true is that 
the ordinary localization functor $L : \CC \to \CC[W^{-1}]$ 
exhibits $\CC[W^{-1}]$ as an $\infty$-categorical 
localization of $\CC$ at $W$. This is a direct consequence of Hinich's 
rectification results \cite[Theorem 4.1.1]{Hinich2} for $\infty$-functors.
The question of whether the ordinary functor $L$ is an $\infty$-categorical 
localization can be addressed by using the explicit 
criteria established in \cite[Key Lemma 1.3.6]{Hinich3}. 
As a side-remark, we would like to note that
reflective localizations of categories are $\infty$-localizations,
hence our strictification theorem for locally covariant AQFTs in one spacetime dimension 
(see Example \ref{ex:explicitlistofexamples}) can be deduced alternatively from 
Theorem \ref{theo:emptystrictification} and Appendix \ref{app:Loc1}.
\end{rem}

\begin{ex}\label{ex:RCE}
Given any globally hyperbolic Lorentzian manifold $M$ with a 
sufficiently small and compactly supported metric perturbation $h$, one can form the RCE category
\begin{flalign}\label{eqn:RCEcategory}
\CC \,:=\, \left(
\parbox{3em}{\xymatrix@R=1.2em@C=1.2em{
&\ar[dl]_-{i_+}M_+ \ar[dr]^-{j_+}&\\
M&&M_h\\
&\ar[ul]^-{i_-} M_-\ar[ur]_-{j_-}&
}}
\right)
\end{flalign}
which controls the relative Cauchy evolution indexed by the pair $(M,h)$.
Here $M_h$ denotes the perturbed spacetime
and $M_\pm := M\setminus J^\mp(\supp(h))$ is the subspacetime with the causal past/future
of the support of $h$ removed. The arrows in \eqref{eqn:RCEcategory} are inclusions
and, by construction, they are all Cauchy morphisms, i.e.\ $W=\Mor\,\CC$. It is 
shown in \cite[Lemma 2.2]{BFSRCE} that the functor
\begin{flalign}
\nn L\,:\, \CC ~&\longrightarrow~\BB\bbZ\quad,\\
\nn M,M_-,M_h,M_+ ~&\longmapsto~\ast\quad,\\
\nn i_-~&\longmapsto~1\quad,\\
j_-,j_+,i_+~&\longmapsto~ 0\quad,
\end{flalign}
defines a localization of $\CC$ at $W$. Furthermore, the result in \cite[Theorem 4.2]{BFSRCE}
implies that \eqref{eqn:Lanadjunction} is a Quillen equivalence for the present example.
(As a side-remark, we would like to note that similar `universal covering' techniques 
for the RCE category as those initiated in \cite[Equation (2.13)]{BFSRCE} can be used to
check directly Hinich's criteria \cite[Key Lemma 1.3.6]{Hinich3}, which by Remark \ref{rem:inftylocalization} 
leads to an independent proof of this theorem.)
This means that the hypotheses of Theorem \ref{theo:emptystrictification} hold true, hence we obtain
a Quillen equivalence
\begin{flalign}
\xymatrix{
L_!\,:\, \mathcal{L}_{\widehat{W}}\AQFT(\CC,\emptyset)\ar@<0.5ex>[r]~&~\ar@<0.5ex>[l]
\Fun\big(\BB \bbZ, \Alg_{\mathsf{As}}^{}(\Ch_\bbK)\big)\,:\,L^\ast
}\quad.
\end{flalign}
This provides a generalization of the strictification theorem in
\cite[Theorem 5.4]{BFSRCE} from the case of linear homotopy AQFTs 
to all $\Ch_\bbK$-valued AQFTs on $\CC$ that satisfy the homotopy time-slice axiom. In particular, 
the strict $\bbZ$-action carried by the object $\bbL L_!(\AAA)\in \Fun\big(\BB \bbZ, \Alg_{\mathsf{As}}^{}(\Ch_\bbK)\big)$
provides a strict realization of the RCE for a $\Ch_\bbK$-valued AQFT $\AAA\in \AQFT(\CC,\emptyset)$ 
that satisfies the homotopy time-slice axiom.
\end{ex}


\section*{Acknowledgments}
We would like to thank the anonymous referees for valuable comments 
that helped us to improve the manuscript.
V.C.\ is partially supported by grants FPU17/01871 and 
PID2020-117971GB-C21 of the Spanish Ministry of Science and Innovation 
and by grant FQM-213 of the Junta de Andaluc\'ia.
A.S.\ gratefully acknowledges the support of 
the Royal Society (UK) through a Royal Society University 
Research Fellowship (URF\textbackslash R\textbackslash 211015)
and Enhancement Awards (RGF\textbackslash EA\textbackslash 180270, 
RGF\textbackslash EA\textbackslash 201051 and RF\textbackslash ERE\textbackslash 210053). 
The work of M.B.\ is fostered by 
the National Group of Mathematical Physics (GNFM-INdAM (IT)).

\section*{Conflict of interest statement}
On behalf of all authors, the corresponding author states that there is no conflict of interest. 

\section*{Data availability statement}
Data sharing not applicable to this article as no datasets were generated or analysed during the current study.

\appendix

\section{\label{app:Loc1}Localization of $\ovr{\Loc_1}$ at Cauchy morphisms}
In this appendix we present a very explicit model
for the localization at all Cauchy morphisms
of the orthogonal category $\ovr{\Loc_{1}}$
of connected\footnote{Restricting to connected manifolds 
simplifies the presentation of this appendix. Our proof below generalizes in a 
fairly obvious way to non-connected manifolds by treating each connected component separately.}
$1$-dimensional spacetimes.
The orthogonality relation in this case is empty $\perp_{\Loc_1}^{} =\emptyset$, and so 
will be the orthogonality relation on the localized category. 
Up to equivalence of (orthogonal) categories, we can present
$\Loc_1$ as the category whose objects are pairs $(M,e)$ consisting
of a $1$-manifold $M\cong \bbR$ and a non-degenerate $1$-form $e\in\Omega^1(M)$ 
(the vielbein, encoding the metric and time-orientation), and whose morphisms
$f:(M,e)\to (M^\prime,e^\prime)$ are all embeddings $f: M\to M^\prime$ such that
$f^\ast(e^\prime) = e$. Note that each object $(M,e)\in\Loc_1$ embeds via a $\Loc_1$-morphism
into $(\bbR,\dd t)$, where by $t$ we denote a (time) coordinate on $\bbR$. As a consequence,
$\Loc_1$ is equivalent to its full subcategory $\Loc^{\mathrm{skl}}_1\subseteq \Loc_1$ consisting of the objects
$((a,b), \dd t)\in \Loc_1$ determined by open intervals $(a,b)\subseteq \bbR$, with $-\infty \leq a<b\leq +\infty$.
(The superscript ${}^{\mathrm{skl}}$ refers to the concept of `skeletal models' as introduced in \cite{Skeletal}.)
The morphisms are given by translations 
$ f_{\xi} : ((a,b), \dd t)\to ((a^\prime,b^\prime), \dd t)\,,~
t\mapsto t+\xi$, with $\xi\in \bbR$ satisfying the condition
$a^\prime - a\leq \xi \leq b^\prime-b$ such that the interval $(a,b)$ 
is mapped into the interval $(a^\prime,b^\prime)$.
\sk

Let us denote by $\BB\bbR^{\delta}$ the category consisting of a single object, say $\ast$, with $\Hom$-set 
the Abelian group $\bbR^{\delta}:=\bbR$ (with no topology attached). The functor
\begin{flalign}
\nn j \,:\, \BB\bbR^{\delta}~&\longrightarrow~\Loc^{\mathrm{skl}}_1\quad,\\
\nn \ast~&\longmapsto~(\bbR,\dd t)\quad,\\
\xi\in \bbR^{\delta}~&\longmapsto~(f_\xi : (\bbR,\dd t)\to (\bbR,\dd t))\quad
\end{flalign}
is fully faithful and it admits a left adjoint given by
\begin{flalign}
\nn L\,:\, \Loc^{\mathrm{skl}}_1~&\longrightarrow~\BB\bbR^{\delta}\quad,\\
\nn ((a,b),\dd t)~&\longmapsto ~ \ast\quad,\\
\big(f_\xi: ((a,b), \dd t)\to ((a^\prime,b^\prime), \dd t)\big) ~&\longmapsto ~\xi\in \bbR^{\delta}\quad.
\end{flalign}
For completeness, let us mention that the adjunction unit $\eta : \id \to j L$
is given by the components $\eta_{(a,b)}:=f_0  : ((a,b),\dd t) \to (\bbR,\dd t)$
and that the counit $\epsilon : Lj\to \id$ is the identity natural isomorphism
of $Lj = \id$. Equipping all these categories with the empty orthogonality relation,
we obtain from Proposition \ref{propo:OCatdetectlocalization} that 
$L : \ovr{\Loc^{\mathrm{skl}}_1} \to\ovr{\BB\bbR^{\delta}}$
is a reflective localization at all morphisms. (Note that every morphism
in $\ovr{\Loc^{\mathrm{skl}}_1}$ is Cauchy.) Choosing a quasi-inverse of the equivalence
$\ovr{\Loc^{\mathrm{skl}}_1} \stackrel{\sim}{\longrightarrow} \ovr{\Loc_1}$ 
of orthogonal categories, we obtain a localization 
$\ovr{\Loc_1} \stackrel{\sim}{\longrightarrow} \ovr{\Loc^{\mathrm{skl}}_1}  
\stackrel{L}{\longrightarrow} \ovr{\BB\bbR^{\delta}}$ of $\ovr{\Loc_1}$ at all (Cauchy) morphisms.
Since equivalences in $2$-categories can be upgraded to adjoint equivalences, 
the latter localization is reflective too.


\end{document}